\documentclass{llncs}
\usepackage[utf8]{inputenc}
\usepackage[english]{babel}
\usepackage[T1]{fontenc}
\usepackage{setspace}
\usepackage{graphicx, caption}
\usepackage{multirow}
\usepackage{hyperref}
\usepackage[dvipsnames, usenames]{color}
\usepackage{amsmath}
\usepackage{subfigure}
\usepackage{amsmath}
\usepackage{amsfonts}
\usepackage{verbatim}
\usepackage{enumitem}
\usepackage{amssymb}
\usepackage{enumitem}
\usepackage{pifont}
\usepackage{sectsty}
\sectionfont{\fontsize{14}{14}\selectfont}
\usepackage{multicol}
\usepackage{wrapfig}
\begin{document}
\pagestyle{plain} 
\captionsetup{font=small, textfont=it}
\title{Multidimensional Dominance Drawings}
\author{Giacomo Ortali \inst{*} \and Ioannis G. Tollis \inst{**}
\\
\inst{*}University of Perugia \email{giacomo.ortali@gmail.com} 
\\
\inst{**}Computer Science Department, University of Crete, Heraklion, Crete, Greece and Tom Sawyer Software, Inc. Berkeley, CA 94707 U.S.A.  \email{tollis@csd.uoc.gr}
\institute{}}
\maketitle 
	\begin{abstract}
	Let $G$ be a DAG with $n$ vertices and $m$ edges. Two vertices $u,v$ are incomparable if $u$ doesn't reach $v$ and vice versa. We denote by \emph{width} of a DAG $G$, $w_G$,  the maximum size of a set of incomparable vertices of $G$. In this paper we present an algorithm that computes a dominance drawing of a DAG G in $k$ dimensions, where $w_G \le k \le \frac{n}{2}$. The time required by the algorithm is $O(kn)$, with  a precomputation time of $O(km)$, needed to compute a  \emph{compressed transitive closure} of $G$, and extra $O(n^2w_G)$ or $O(n^3)$ time, if we want $k=w_G$. Our algorithm gives a tighter bound to the dominance dimension of a DAG.  As corollaries, a new family of graphs having a 2-dimensional dominance drawing and a new upper bound to the dimension of a partial order are obtained. We also introduce the concept of transitive module and dimensional neck, $w_N$, of a DAG $G$ and we show how to improve the results given previously using these concepts.
\end{abstract}
\section{Introduction}
Dominance drawings of directed acyclic graphs (DAGs) are very important in many areas of research, including graph drawing~\cite{DBLP:conf/cccg/ElGindyHLMRW93}, computational geometry~\cite{DBLP:journals/dcg/BattistaTT92}, information visualization\cite{10.1007/3-540-45848-4_4}, even in very large databases~\cite{DBLP:conf/edbt/VelosoCJZ14,Zhou:2017:DRF:3035918.3035927}, just to mention a few.  They combine the aspect of drawing a DAG on the grid with the fact that the transitive closure of the DAG is apparently obvious by the dominance relation between grid points associated with the vertices.  In other words, in a dominance drawing a vertex $v$ is reachable from a vertex $u$ if and only if all the coordinates of $v$ are greater than or equal to the coordinates of $u$ in $\Gamma$.  In a DAG $G$ with $n$ vertices and $m$ edges two vertices $u,v$ are incomparable if $u$ doesn't reach $v$ and vice versa. We denote by \emph{width} of a DAG $G$, $w_G$, the maximum size of a set of incomparable vertices of $G$.  Notice that it is not possible to find dominance drawings in 2-dimensions for most DAGs. The smallest number $d$ for which a given DAG $G$ has a $d$-dimensional dominance drawing is called \emph{dominance drawing dimension}, denoted by $d_G$, and it is a known NP-hard problem to compute it~\cite{Yannakis}.  In this paper we present algorithms for computing a $k$-dimensional dominance drawing of $G$, where $k\ge d_G$.  Our algorithms are efficient and are based on various decomposition techniques on the DAG.

In 2-dimensions the dominance drawing method for planar DAGs has many important aesthetic properties, including small number of bends, good vertex placement, and symmetry display~\cite{DBLP:journals/dcg/BattistaTT92,st-planar}.  A 2-dimensional dominance drawing $\Gamma$ of a planar DAG $G$ can be computed in linear time, such that for any two vertices $u$ and $v$ there is a directed path from $u$ to $v$ in $G$ if and only if $x(u) \le x(v)$ and $y (u) \le y (v)$ in $\Gamma$~\cite{DBLP:journals/dcg/BattistaTT92,st-planar}. Since most DAGs have dominance dimension higher that two, the concept of  weak dominance drawings was introduced in~\cite{DBLP:journals/corr/abs-1108-1439,DBLP:conf/gd/KornaropoulosT12a}. This concept has many applications including the drawing of DAGs in the overloaded orthogonal model~\cite{DBLP:journals/jgaa/KornaropoulosT16}.  In weak dominance, for any two vertices $u$ and $v$ if there is a directed path from $u$ to $v$ in $G$ then $x(u) \le y(v)$ and $y (u) \le y (v)$ in $\Gamma$.  However, the reverse does not hold.  Hence, we have a falsely implied path (fip) when $x(u) \le y(v)$ and $y (u) \le y (v)$, but there is no path from $u$ to $v$.  Kornaropoulos and Tollis~\cite{DBLP:journals/corr/abs-1108-1439} proved that the problem of minimizing the number of fips is NP-hard and gave some upper bounds on the number of fips.

Several researchers adopted the concept of weak dominance drawing in order to construct a compact representation of  the reachability information of very large graphs that are produced by very large datasets in the database community~\cite{DBLP:conf/edbt/VelosoCJZ14}.   Li, Hua, and Zhou considered high dimensional dominance drawings in order to reduce the number of fips and obtain efficient solutions to the reachability problem~\cite{DBLP:journals/www/LiHZ17}. Namely, they use high dimensional dominance drawings in order to reduce the number of fips and describe heuristics to obtain a system that resolves reachability queries in linear (or constant) time as demonstrated by their experimental work~\cite{DBLP:journals/www/LiHZ17}.  

In this paper we present an algorithm, called kD-Draw, that computes a dominance drawing of a DAG G in $k$ dimensions, where $w_G \le k \le \frac{n}{2}$. The time required by the algorithm is $O(kn)$, with a precomputation time of $O(km)$, needed to compute a "special" transitive closure of $G$, called \emph{compressed transitive closure}.  If we want to have $k=w_G$ then an extra $O(n^2w_G)$ or $O(n^3)$ time is required to find a proper decomposition of $G$. Algorithm kD-Draw defines a new upper bound to the dominance dimension of a DAG.  As corollaries we obtain (a) a new family of graphs that admit a 2-dimensional dominance drawing and (b) a new upper bound to the dimension of a partial order. We also introduce the concepts of transitive modules and dimensional neck $w_N$ of a DAG $G$ and we show how to use them in order to  improve our results using these concepts.

Our paper is structured as follows: In Section~\ref{Section:Preliminaries} we describe necessary preliminary results. In Section~\ref{Section:Multidimensional} we  first introduce the new dominance drawing technique for 2 dimensions, and then these results are extended to $k$ dimensions. We introduce a new upper bound to the dominance dimension of a DAG and we discuss some additional implications of the results of this section. In Section~\ref{Section:Modules} we introduce the concept of transitive modules and we use it to improve the upper bound on the number of dimensions presented in the previous section. In Section \ref{Section:Conclusion} we present our conclusions and we discuss interesting open problems that naturally arise from our two different approaches to $k$-dimensional dominance drawing.
\section{Preliminaries}
\label{Section:Preliminaries}
Let $G=(V,E)$ be an directed acyclic graph (DAG) with $n$ vertices and $m$ edges. An \emph{st-graph} is a DAG with one source $s$ and one sink $t$. In order to simplify our presentation, for the rest of the paper we will assume that every DAG is an st-graph. We do it without loss of generality, since we can obtain an st-graph from any DAG by adding a virtual source and a virtual sink and connecting them to all sources and sinks, respectively.  Testing if $G$ has dominance drawing dimension 2 requires linear time~\cite{McConnell}, while testing if its dimension is greater than or equal to 3 is NP-complete~\cite{Yannakis}.  An efficient algorithm to compute 2-dimensional dominance drawings for planar st-graphs is shown in~\cite{st-planar}. 
A partial order is a mathematical formalization of the concept of ordering. Any partial order $P$ can be viewed as a transitive DAG. The results obtained for DAGs and their dominance drawing dimension transfer directly to partial orders and their dimension and vice-versa. Hence, we can talk about the results known for partial orders and for DAGs with no distinction. 
In~\cite{N/2} Hiraguchi proved a theorem that gives a tight upper bound on the dominance dimension of $G$, which is $\frac{n}{2}$, as shown in the following lemma~\cite{N/2} (for a different proof, see~\cite{DBLP:journals/dm/Bogart73}):
\begin{lemma}
	\label{lemma:n/2}
	The dominance dimension of an st-graph $G$ having $n$ vertices is at most  $\frac{n}{2}$. In other words: $d_G \le \frac{n}{2}$.
\end{lemma}

Now we introduce the concept of \emph{channel}, which is a generalization of the concept of \emph{path}. Then we will introduce a graph decomposition into channels, called \emph{channel decomposition}. This decomposition will be used in the next section to improve the upper bound stated in Lemma~\ref{lemma:n/2}.

A \emph{channel} $C$ is an ordered set of vertices such that, given any two vertices $v,w\in C$, $v$ precedes $w$ in the order of channel $C$ if and only if $w$ is reachable from $v$ in $G$. If $u$ precedes $v$ in the order of a channel $C$ then $v$ is a successor of $u$ in $C$. We denote by \emph{channel decomposition} of $G$ a set of channels $S_c=\{C_1,...,C_k\}$ so that the source $s$ and the sink $t$ of $G$ are contained in every channel and every other vertex of $G$ is contained in exactly one channel. The number of channels of a decomposition $S_c$ is called \emph{size} of $S_c$.

Figure~\ref{fi:channeldec} shows an st-graph and a minimum size channel decomposition of it. The source and the sink of the st-graph $G$ depicted in Part~(a) are respectively 0 and 15. Red edges of Part~(b) connect two consecutive vertices of a same channel. The dashed red edges represent edges of the transitive closure of $G$ that do not belong to $G$. The channel decomposition showed in Part~(b) is $S_c=\{C_1,C_2,C_3,C_4\}$, where: $C_1=\{0,1,4,5,12,13,15\}$; $C_2=\{0,3,7,11,15\}$; $C_3=\{0,2,6,10,14,15\}$; $C_4=\{0,8,9,15\}$.  Notice that this channel decomposition is minimum, since the width of $G$ is four.  We will revisit this graph and will show a different dominance drawing obtained using the concept of transitive modules that we will introduce in Section~\ref{Section:Modules}.
\begin{figure}[ht]
	\centering
	\subfigure[]
	{\includegraphics[width=0.3\linewidth]{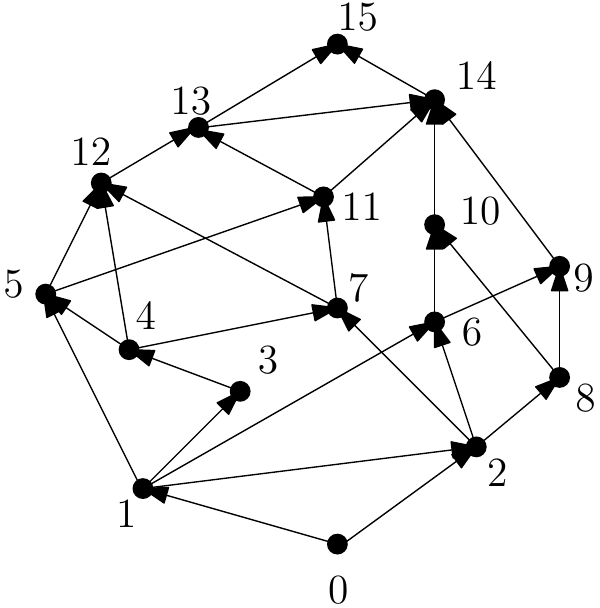}}
	\hfil
	\subfigure[]
	{\includegraphics[width=0.35\linewidth]{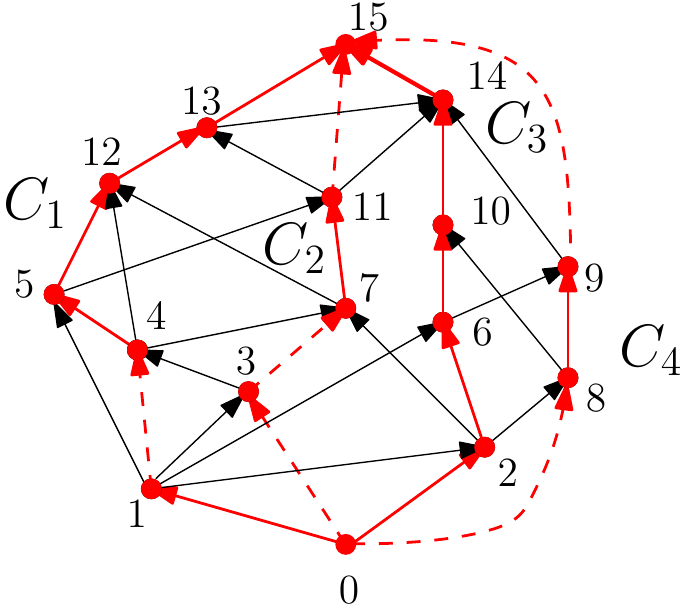}}
	\caption{(a) An st-graph and a channel decomposition of it, with source 0 and sink 15.  (b) A minimum size channel decomposition of $G$ is $S_c=\{C_1,C_2,C_3,C_4\}$, where: $C_1=\{0,1,4,5,12,13,15\}$; $C_2=\{0,3,7,11,15\}$; $C_3=\{0,2,6,10,14,15\}$; $C_4=\{0,8,9,15\}$.}
	\label{fi:channeldec}
\end{figure}
\subsection{Minimum Size Channel Decomposition}\label{se:minimumdecomposition}
In this section we introduce the concept of width of an st-graph and we give a short description of an algorithm that computes a minimum size channel decomposition of an st-graph. Two vertices $u,v\in V$ are incomparable if $u$ doesn't reach $v$ and vice versa. We denote by \emph{width} of a DAG $G$, $w_G$,  the maximum size of a set of incomparable vertices of $G$. Computing the width of a graph requires linear time~\cite{DBLP:journals/tods/Jagadish90}. Additionally, the following result is proved in~\cite{dilworth}:
\begin{lemma}
	\label{label:width}
	The minimum size of a channel decomposition of $G$ is equal to the width $w_G$ of $G$.
\end{lemma}
Jagadish~\cite{DBLP:journals/tods/Jagadish90} presented an algorithm to compute a channel decomposition with the minimum number of channels in $O(n^3)$ time. To make our paper self-contained, in the next paragraph we outline a simple variation of this algorithm.

First we compute a graph $G'$ from $G$ such that: (a) any vertex $v$ of $G$ is associated with exactly two vertices of $G'$, that we call $x_v$ and $y_v$; (b) $(x_v,y_v)$ is an edge of $G'$ for any couple of vertices $x_v,y_v$; and (c) each edge $(u,v)$ of $G$ is associated with an edge $(y_u,x_v)$ of $G'$. Any channel $C=(s,v_i,...,v_j,t)$ starting from the source $s$ and ending at the sink $t$ of $G$ corresponds to a channel $C'=(x_s,y_s,x_{v_i},y_{v_i},...,x_{v_j},y_{v_j},x_y,y_t)$ of $G'$. Now we solve the standard max-flow problem on the acyclic graph $G'$ by techniques such as~\cite{10021311931}. In order to find a channel $C$ of $G$ from the flow of $G'$ we firstly add  vertex $s$ as the first element of $C$ and we decrease the flow sent through edge $(x_s,y_s)$ by one; then, if $v$ is the last vertex added to $C$, we look for an edge $(y_v,x_u)$ having a positive flow, we decrease its flow by one and we add $u$ as the last element of $C$. We repeat this process until $t$ is added to $C$. We find channels until every vertex of $G$ belongs to at least one channel. A vertex could be inserted to more than one channels; in this case we simply remove it from all the channels it belongs to, except for one. The computed channels constitute a minimum size channel decomposition. 

A faster algorithm, that runs in in $O(w_Gn^2)$ time, to compute a channel decomposition with the minimum number of channels is presented in~\cite{chendag}. We call this algorithm "Algorithm Channels-Generation". Hence, we have the following lemma:
\begin{lemma}\label{lemma:channeldecomposition}
	Algorithm Channels-Generation computes a channel decomposition of a DAG $G$ having $w_G$ channels in $O(w_Gn^2)$ time.
\end{lemma}
\subsection{Projections and Compressed Transitive Closure}
\label{Section:projections}
Now we are ready to introduce the concept of \emph{projection} of a vertex $v$ on a channel $C$. We will also briefly talk about a data structure, called \emph{compressed transitive closure}, that can be used to store  all the projections for any vertex efficiently. The projections of a vertex $v$ will be used to decide its coordinates in all the dimensions of our multidimensional dominance drawing.

We denote by $u=(i,j)$ the fact that $u$ is the jth vertex of channel $C_i$. By the definition of channel decomposition we have $t=(i,|C_i|)$ and $s=(i,0)$ for any $i\in [1,k]$. We denote by \emph{projection} of a vertex $u\in V$ on a channel $C\in S$ the vertex $v\in C$ having the lowest position in $C$ among all the vertices of $C$ reachable from $u$. We denote it by $Proj_C (u) = v$, if $u\in C$ then $Proj_C (u) = u$. Notice that all the vertices can reach at least a vertex of every channel, since the sink $t$ of the graph belongs to all the channels. Hence, the projection $uC$ is defined for any couple $(u,C)\in (V, S_c)$.
	
The following lemma is immediate by the definition of projection:
\begin{lemma}
	\label{lemma:TC}
	Let $u\in V$ be a vertex and let $Proj_{C_i} (u) = v = (i,j)$ be the projection of $u$ on $C_i$; $u$ can reach a vertex $v'=(i,j')\in C_i$ if and only if $j'\ge j$.
\end{lemma} 

Lemma~\ref{lemma:TC} shows that we can use the projections to study the reachability properties of $G$. Jagadish describes in~\cite{DBLP:journals/tods/Jagadish90} a data structure that we can use to store all the projections of any vertex of $G$ efficiently. This data structure is called \emph{compressed transitive closure}. In the same paper he shows how to compute the compressed transitive closure of $G$ in $O(km)$ time, store it in $O(kn)$ space and use it to read any projection $Proj_C (u) = v$ for any vertex $u$ and any channel $C$ in constant time.

Figure~\ref{fi:projections} shows a representation of the compressed transitive closure of graph $G$ depicted in Figure~\ref{fi:channeldec}(a) given a channel decomposition of it, which is depicted in Figure~\ref{fi:channeldec}(b). In order to have a better visualization, we restrict the representation of the compressed transitive closure to only two channels. Figure~\ref{fi:projections}(a) shows the compressed transitive closure restricted to channels $C_1$ and $C_2$. The red edges connect vertices of the same channel. Let $u$ and $v$ be two vertices not belonging to the same channel. A gray edge $(u,v)$ indicates that $v$ is the projection of $u$ on the channel of vertex $v$. Part~(b), Part~(c), Part~(d), Part~(e), Part~(f) represent the compressed transitive closure for the pairs of channels:  $(C_1,C_3)$, $(C_1,C_4)$, $(C_2,C_3)$, $(C_2,C_4)$ and $(C_3,C_4)$, respectively.
\begin{figure}
	\centering
	\subfigure[]
	{\includegraphics[width=0.3\linewidth]{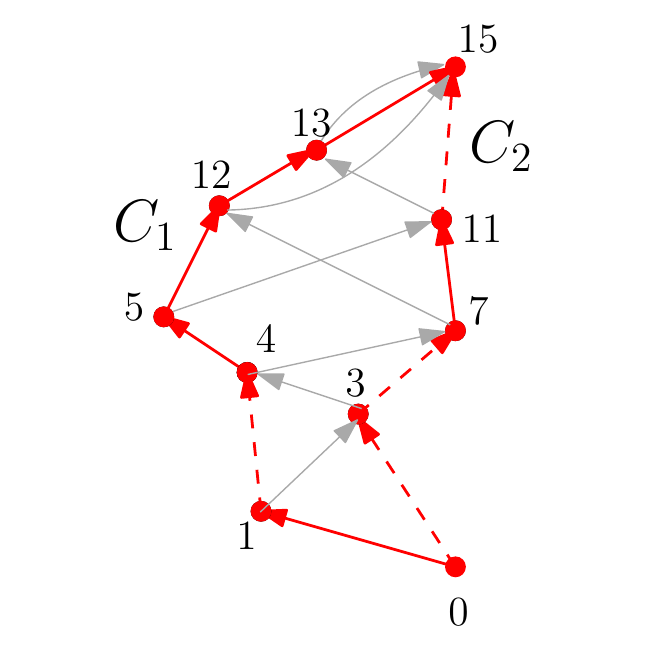}}
	\subfigure[]
	{\includegraphics[width=0.3\linewidth]{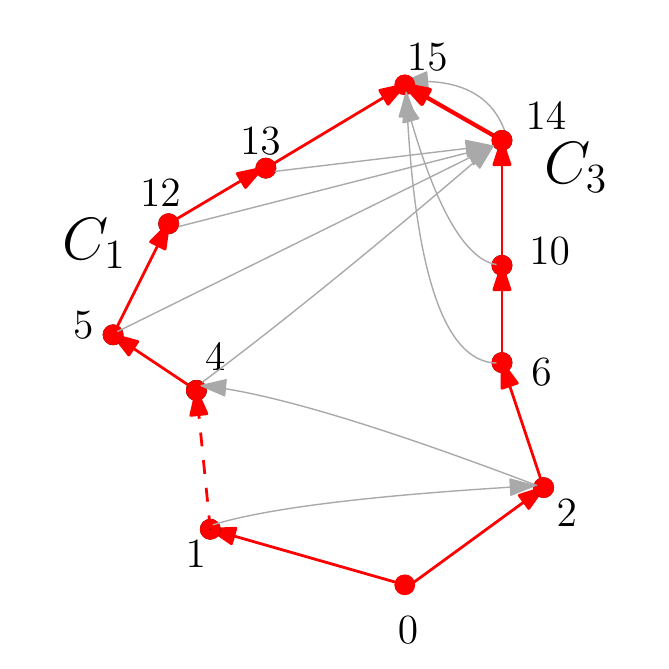}}
	\subfigure[]
	{\includegraphics[width=0.3\linewidth]{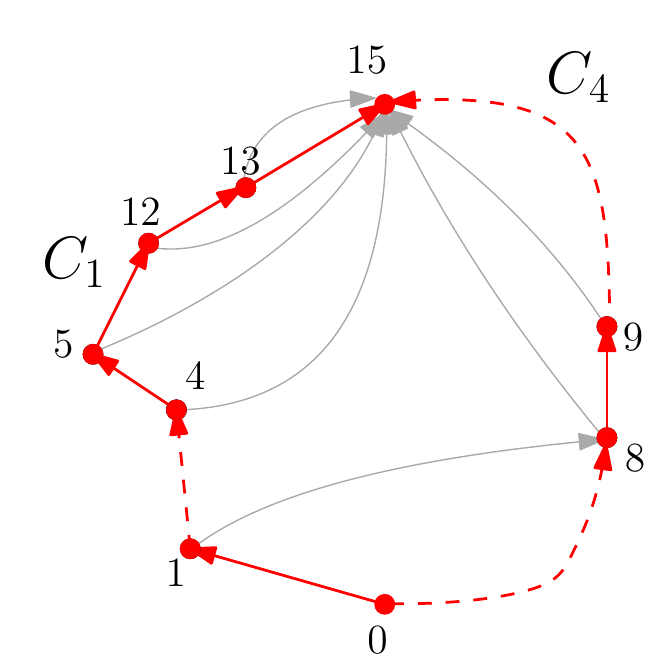}}
	\subfigure[]
	{\includegraphics[width=0.3\linewidth]{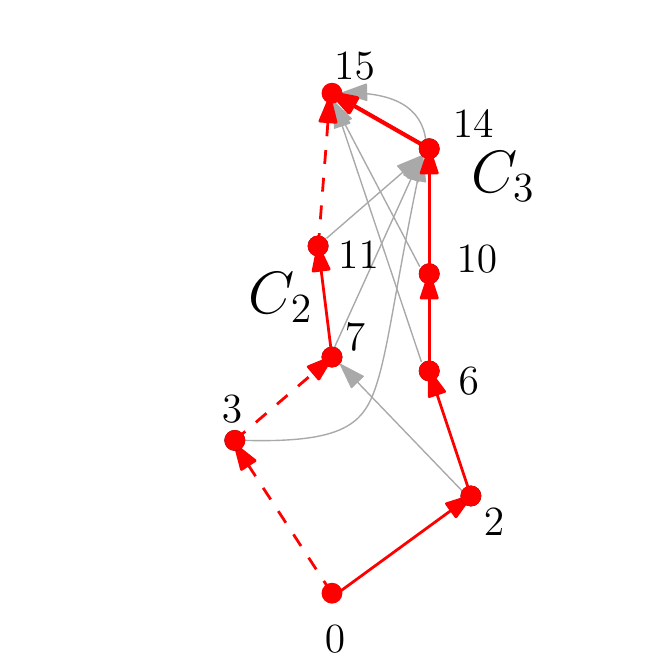}}
	\subfigure[]
	{\includegraphics[width=0.3\linewidth]{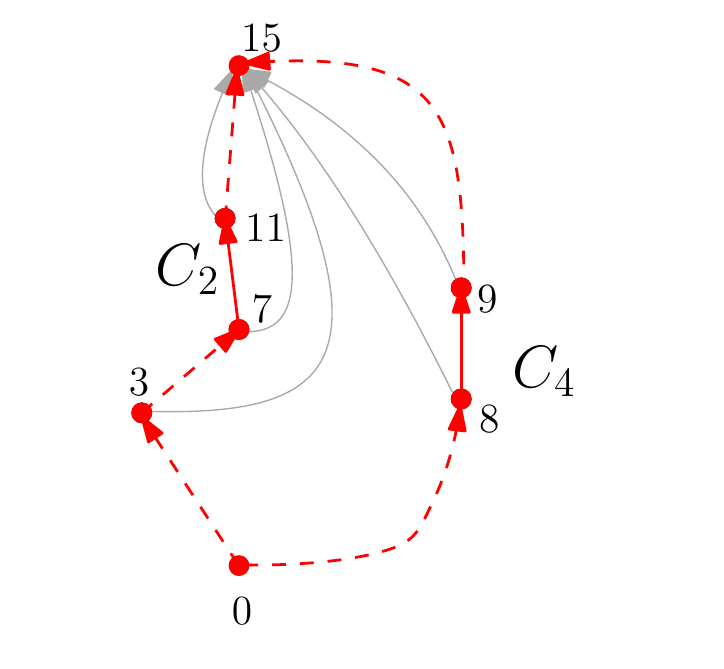}}
	\subfigure[]
	{\includegraphics[width=0.3\linewidth]{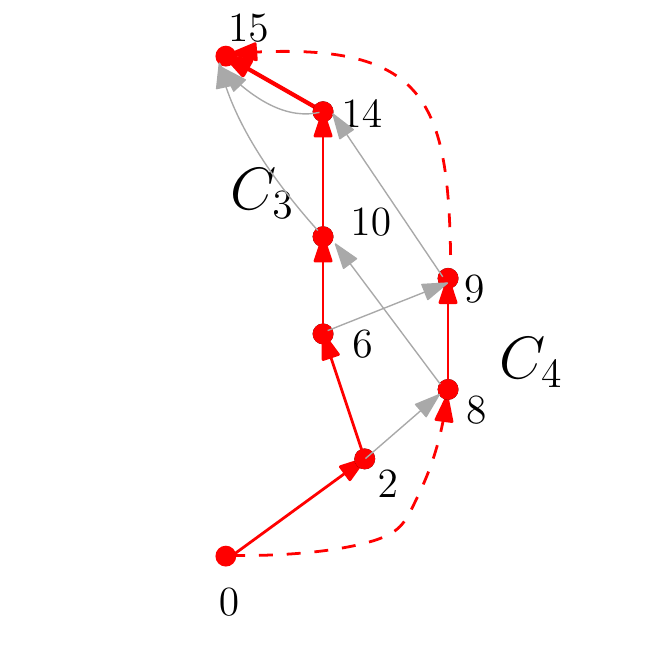}}
	\caption{Representation of the compressed transitive closure of graph $G$ depicted in Figure~\ref{fi:channeldec}(a) given the channel decomposition of it that is depicted in Figure~\ref{fi:channeldec}(b).  In order to have a clearer figure, we restrict the representation of the compressed transitive closure by showing pairs of channels at a time.}
	\label{fi:projections}
\end{figure}

\section{Multidimensional Dominance Drawing}
\label{Section:Multidimensional}
Let $S_c=\{C_1,...,C_k\}$ be a channel decomposition of size $k$ of st-graph $G=(V,E)$.
In this section we explain how we can use $S_c$ and the projections of the vertices of $G$ in order to create a dominance drawing of $G$ in $k$ dimensions. 

In Subsection \ref{subsection:2} we present Algorithm 2-Dimensional-Draw (or simply 2D-Draw), that, for $k=2$, computes a 2 dimensional dominance drawing of $G$. In Subsection \ref{subsection:k} we extend the 2-dimensional algorithm by introducing Algorithm k-Dimensional-Draw (or simply kD-Draw), that computes a $k$ dimensional dominance drawing of $G$ for any channel decomposition of size $k$.

\subsection{Base case: 2 Channels}
\label{subsection:2}
Let $S_c=\{C_1,C_2\}$ be a channel decomposition of an st-graph $G=(V,E)$ of size 2. We will present Algorithm 2D-Draw that receives as input $G$ and $S_c$ and produces a two dimensional dominance drawing $\Gamma$ of $G$.

The algorithm uses the order of the vertices in each channel, $C_1$ and $C_2$, in order to assign $X$ and $Y$ coordinates to the vertices of $C_1$ and $C_2$. Then the algorithm assigns appropriate $X$ coordinates to the vertices that do not belong to channel $C_1$ and $Y$ coordinates to the vertices that do not belong to channel $C_2$. It does this by assigning an $X$-coordinate to each vertex in $C_2$ using the corresponding projection's $X$-coordinate, that was already assigned before (as shown in Line 4). Similarly, it assigns a $Y$-coordinate to each vertex in $C_1$ using the corresponding projection $Y$-coordinate, that was already assigned before (as shown in Line 6). This process is shown in Lines 7-12.\\\\
	\textbf{Algorithm} 2D-Draw($G$, $S_c$)\\
	1. \indent $\Gamma$ = new 2-dimensional drawing\\
	2. \indent \textbf{For} any $v=(i,j)\in V$:\\
	3. \indent \indent \textbf{If}($i=1$):\\
	4. \indent \indent \indent $X(v)=j$\\
	5. \indent \indent \textbf{Else}:\\
	6. \indent \indent \indent $Y(v)=j$:\\
	7. \indent \textbf{For} any $v\in C_1$\\
	8. \indent \indent $Proj_{C_2} (v)=(2,j)$\\
	9. \indent \indent $Y(v)=j$\\
	10.\indent \textbf{For} any $v\in C_2$\\
	11.\indent \indent $Proj_{C_1} (v)=(1,j)$\\
	12.\indent \indent $X(v)=j$\\
	13.\indent \emph{output:}  $\Gamma$\\\\
	
Figure~\ref{fi:2d} is an illustration of Algorithm 2D-Draw. Part~(a) shows an st-graph $G$. The source of $G$ is vertex 0 and the sink of $G$ is vertex 7. Part~(b) shows a channel decomposition $S_c=\{C_1,C_2\}$ of $G$, where: $C_1=\{0,3,4,7\}$; $C_2=\{0,1,2,5,6,7\}$. Part~(c) shows the projections of the vertices of $G$ on the channels they don't belong to. Part~(d) shows the $X$ coordinate assignment for the vertices of $C_1$ and the $Y$ coordinate assignment for the vertices of $C_2$. These assignments are shown by writing the number of a vertex next to the corresponding coordinate; this is performed in Lines 2-6 of Algorithm 2D-Draw. Part~(e) shows the assignment of the other coordinate to every vertex of $G$, that performed by Algorithm 2D-Draw in Lines 7-12 by using the projections of the vertices. For example, The projection of vertex 1 of channel $C_2$ is $4$, hence, $X(1)=X(4)$.
\begin{figure}[ht]
	\centering
	\subfigure[]
	{\includegraphics[width=0.25\linewidth]{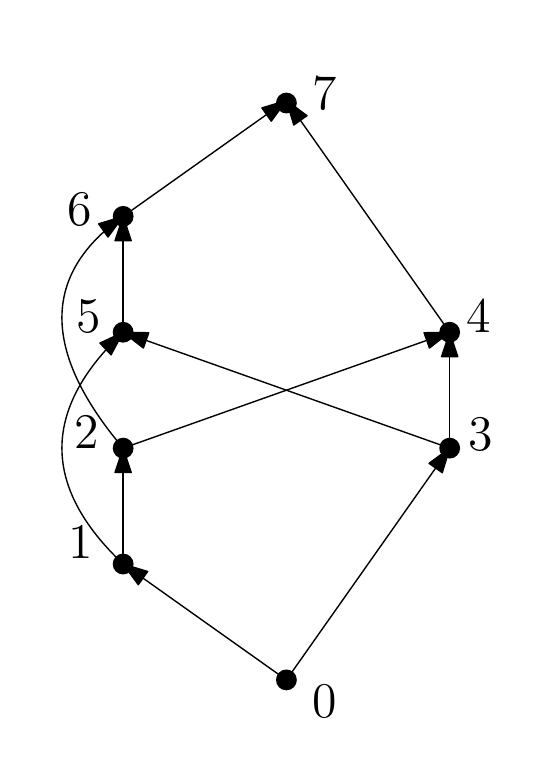}}
	\hfil
	\subfigure[]
	{\includegraphics[width=0.25\linewidth]{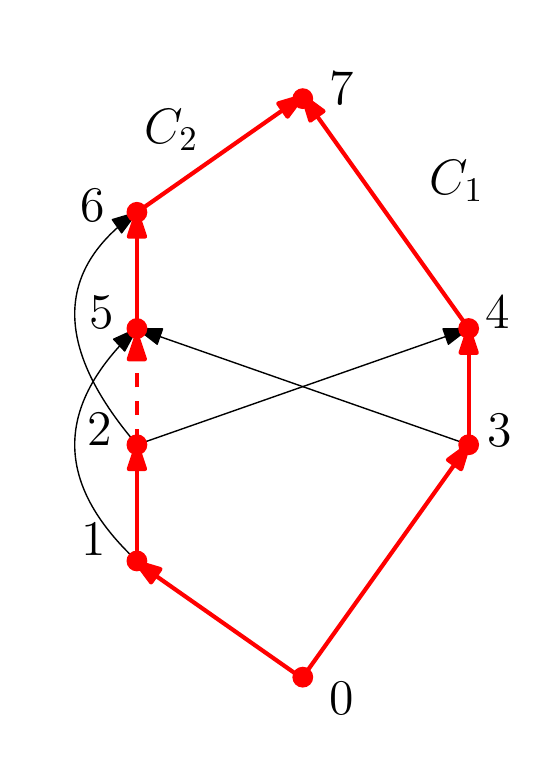}}
	\hfil
	\subfigure[]
	{\includegraphics[width=0.25\linewidth]{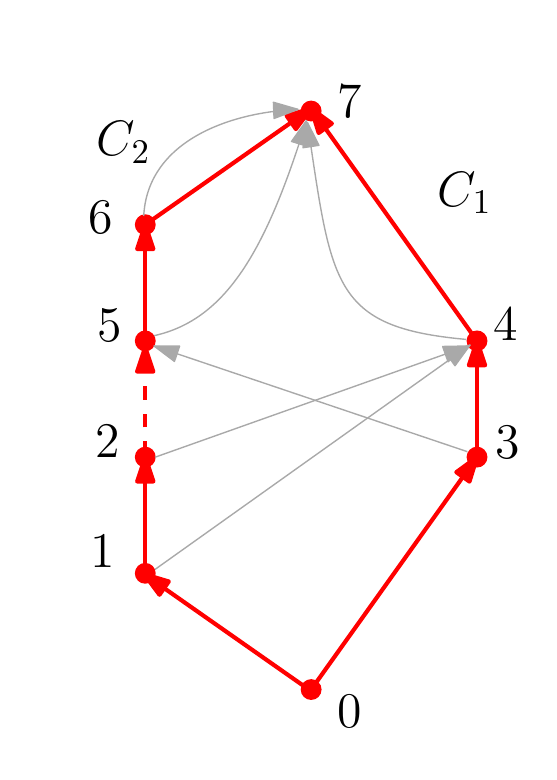}}
	\\
	\hfil
	\subfigure[]
	{\includegraphics[width=0.25\linewidth]{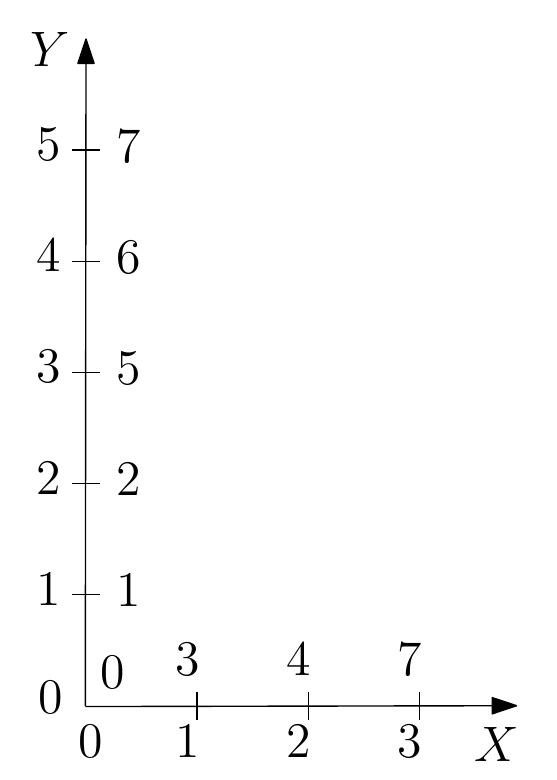}}
	\hfil
	\subfigure[]
	{\includegraphics[width=0.25\linewidth]{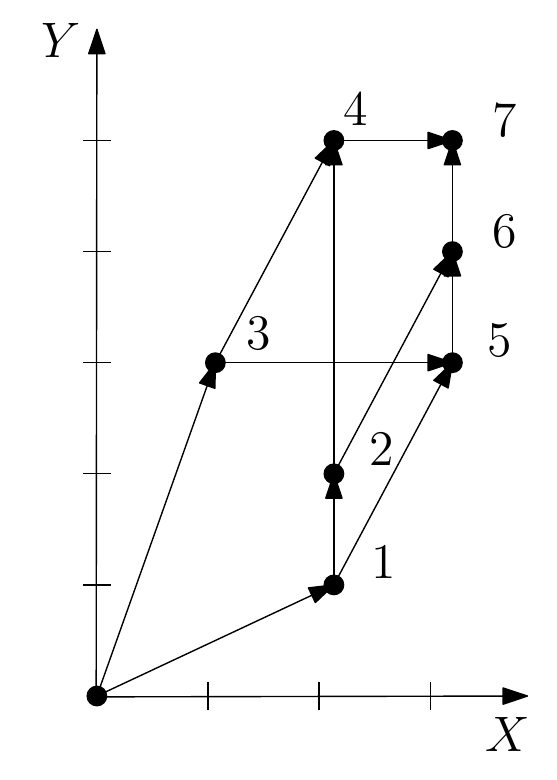}}
	\caption{Illustration of Algorithm 2D-Draw: (a) an st-graph $G$;  (b) a two-channel decomposition $S_c=\{C_1,C_2\}$ of $G$, where: $C_1=\{0,3,4,7\}$; $C_2=\{0,1,2,5,6,7\}$;  (c) the projections of the vertices of $G$ on the channels they do not belong to; (d) the $X$ coordinate assignment for the vertices of $C_1$ and the $Y$ coordinate assignment for the vertices of $C_2$; (e) the assignment of the other coordinate of every vertex of $G$.}
	\label{fi:2d}
\end{figure}
\indent
\\
\indent
We denote by $r(u,v)=yes$ the fact that there exists a path starting from $u$ and ending at $v$ in $G$ and by $r(u,v)=no$ the fact that this path does not exist. Let $\Gamma$ be a drawing of $G$. We denote by $u\preceq v$ the fact that all the coordinates of $u$ are less than or equal to the coordinates of $v$ in all dimensions of $\Gamma$.
The main result of this subsection is that Algorithm 2D-Draw computes a 2-dimensional dominance drawing of $G$. In other words, we prove that, for any $u,v\in G$: $r(u,v)=yes\Leftrightarrow u\preceq v$. This result will be proved in Lemma~\ref{lemma:dd2}. Before we are able to prove that however, we need the following three lemmas (Lemma~\ref{lemma:same_position2}, Lemma~\ref{lemma:source_sink_2} and Lemma\ref{lemma:reachability_2}) that will be used in the proof of Lemma~\ref{lemma:dd2}.
 
\begin{lemma}
	\label{lemma:same_position2}
	Any two distinct vertices $v,w\in V$ are placed on distinct points in $\Gamma$. 
\end{lemma}
\begin{proof}
	Let $v=(1,j)$ and $u=(h,l)$ be two different vertices. Without loss of generality we suppose that $v\in C_1$ . We have two cases: (1) $u\in C_1$; (2) $u\in C_2$. We need to prove that $X(v)\not = X(u)$ or $Y(v)\not = Y(u)$:
	\begin{enumerate}[leftmargin=*]
		\item $h=1$ and $u=(1,l)$. We have $X(v)=j\not =l = X(u)$, since in this case $j=l$ would imply that $v$ and $u$ are the same vertex.
		\item $h=2$ and $u=(2,l)$. Let $v'=Proj_{C_2}=(2,l')$ and $u'=Proj_{C_1}=(1,j')$ be the projections of $v$ and $u$. According to the algorithm $X(v)=j, Y(v)=l'$ and $X(u)=j', Y(u)=l$. The vertices $v$ and $u$ would be placed in the same point if and only if $l=l'$ and $j=j'$, thus $v=u'$ and $u=v'$. In this case $r(v,u)=yes$ and $r(u,v)=yes$ by definition of projection, so there must be a cycle in $G$. This is a contradiction, since $G$ is a DAG. Consequently $X(v)\not = X(u)$ or $Y(v)\not= Y(u)$.
	\end{enumerate}
\end{proof} 
\begin{lemma}
	\label{lemma:source_sink_2}
	Let $v$ be a vertex of $G$. The source $s$ is dominated by $v$ and the sink $t$ dominates $v$ in $\Gamma$. In other words: $X(s)\le X(v)\le X(t) $ and $Y(s)\le Y(v) \le Y(t)$.  
\end{lemma}
\begin{proof}
	Algorithm 2D-Draw places $s$ such that $X(s)=Y(s)=0$ and it places $t$ such that $X(t)=|C_1|$ and $Y(t)=|C_2|$. Moreover, it places any $v$ such that $0\le X(v)\le |C_1|$ and $0\le Y(v)\le |C_2|$.
\end{proof}
\begin{lemma}
	\label{lemma:reachability_2}
	Let $u$ and $v$ be two vertices of $G$. If $u\in C_1$: $r(u,v)=yes \Leftrightarrow X(u)\le X(v)$. Else, if $u\in C_2$: $r(u,v)=yes \Leftrightarrow Y(u)\le Y(v)$.
\end{lemma}
\begin{proof}
	If $u\in C_1$ Algorithm 2D-Draw places $v$ such that $X(v)=j$, where $Proj_{C_1} (v)=(1,j)$. If $u\in C_2$ Algorithm 2D-Draw places $v$ such that $Y(v)=j$, where $Proj_{C_2} (v)=(2,j)$. The proof is given by Lemma~\ref{lemma:TC}.
\end{proof}
Now, we are ready to prove Lemma~\ref{lemma:dd2}, which states that Algorithm 2D-Draw computes a dominance drawing.

\begin{lemma}
	\label{lemma:dd2}
	Drawing $\Gamma=$ 2D-Draw($G$, $S_c$), computed by Algorithm 2D-Draw is a dominance drawing.
\end{lemma}
\begin{proof}
	Let $v=(1,j)$ and $u=(h,l)$ be two vertices, and assume without loss of generality, that $v\in C_1$. If $u$ or $v$ is the source or the sink the theorem is true by Lemma~\ref{lemma:source_sink_2}. Recall that by Lemma~\ref{lemma:same_position2} Algorithm 2D-Draw never places two distinct vertices on the same point of $\Gamma$.  Hence, in order to prove that $\Gamma$ is  a dominance drawing we have to prove: $u\preceq v$ $\Leftrightarrow$ $r(u,v)=yes$.  By Lemma~\ref{lemma:reachability_2} we have that: $u\preceq v$ $\Rightarrow$ $r(u,v)=yes$. Thus we have to prove the following: $u\preceq v$ $\Leftarrow$ $r(u,v)=yes$. For that, suppose that $u$ is reachable from $v$. Let $u', v'$ be the projections of $u, v$ on the channel they don't belong to, respectively. If $r(u,v)=yes$ we have two cases: (1) $v$ and $u$ are in the same channel or (2) they belong to different channels.
	\begin{enumerate}
		\item Let $v'=Proj_{C_2} (v)=(2,j')$ and $u'=Proj_{C_1} (u)=(2,l')$ be the projections of $v$ and $u$. We have that $l<j$, since $r(u,v)=yes$ and by definition of channels, therefore $X(u)=l<j=X(v)$. If $u'$ is a successor of $v'$ then there is a cycle, therefore $u'=v'$ or $v'$ is a successor of $u'$. This implies that $l'\le j'$, which implies that $Y(u)\le Y(v)$, since $Y(v)=j'$ and $Y(u)=l'$. Hence, $u\preceq v$.
		\item  In this case $v=(1,j)$, $u=(2,l)$, $u'=Proj_{C_1} (u)=(1,l')$ and $v'=Proj_{C_2} (v)=(2,j')$. We have $X(v)=j, X(u)=l'$ and $r(u,v)=yes$. By Lemma~\ref{lemma:reachability_2} we conclude $X(u)\le X(v)$. We have $Y(v)=j'$ and $Y(u)=l$. Additionally, if it were $Y(v)=j'<l=Y(u)$ then $r(v,u)$ by Lemma~\ref{lemma:reachability_2}. In this case $G$ has a cycle, which is a contradiction. Since by hypothesis $r(u,v) = yes$ therefore $Y(u)=l\le l'=Y(v)$ and $u\preceq v$.
	\end{enumerate}
\end{proof} 
\subsection{Generalization: from 2 to k Channels}
\label{subsection:k}
In this subsection we extend the result obtained above by showing how to construct dominance drawings in $k$ dimensions. The algorithm that we present is called Algorithm kD-Draw. The input to the algorithm is a graph $G$ and a channel decomposition with $k$ channels, $S_c=\{C_1,...,C_k\}$. The output of the algorithm is a dominance drawing $\Gamma$ of $G$ in $k$ dimensions.
Similar to the two dimensional case, the $k$-dimensional algorithm uses the order of the vertices in each channel, $C_1,...,C_k$, in order to assign coordinates, $D_1,...,D_k$. Clearly, in the previous section we had $k=2$, and the dimensions were called $D_1=X$ and $D_2=Y$. 

Given any vertex $v=(i,j)\in C_i$, Algorithm kD-Draw will assign the ith dimension of $v$ as $D_i(v)=j$. Then the algorithm assigns appropriate coordinates $D_h(v)$, for all $h\not=i$, as follows: $D_h(v) = l$, where $l$ is the position of the corresponding projection of $v$ in channel $C_h$, i.e., $u=Proj_{C_h} (v)=(h,l)$. In other words, the position of $u$ in channel $C_h$. This is done for all vertices and all dimensions/channels.\\\\
	\textbf{Algorithm} kD-Draw($G$,$S_c=\{C_1,...,C_k\}$)\\
	1. \indent $\Gamma$ = new k-dimensional drawing\\
	2. \indent \textbf{For} any $v=(i,j)\in V$:\\ 
	3. \indent \indent $D_i(v)=j$\\
	4. \indent \textbf{For} any $v\in V$:\\
	5. \indent \indent \textbf{For} any $C_h\in S_c$ such that $v\not \in C_h$:\\
	6. \indent \indent \indent $Proj_{C_h} (v)=(h,l)$\\
	7. \indent \indent \indent $D_{h}(v)=l$\\
	8. \indent \emph{output:}  $\Gamma$\\\\
Figure~\ref{fi:kd} shows an illustration of Algorithm kD-Draw. The input of the algorithm is the same graph $G$ as depicted in Figure~\ref{fi:channeldec}(a) and the channel decomposition $S_c$ of $G$ as depicted in Figure~\ref{fi:channeldec}(b). In this case $k=4$, hence, Algorithm kD-Draw will produce a 4-dimensional dominance drawing of $G$. Part~(a) shows the initialization step, where the algorithm assigns the value of the coordinate in the dimension $D_i$ for every vertex $v\in C_i$ ($0\le i\le 4$). These assignments are shown by writing the number of a vertex next to the corresponding coordinate; this operation is performed in Lines~2-3 of the algorithm. Part~(b) shows the assignment of the other coordinates. The vertex placement is performed by kD-Draw in Lines~4-7 by using the projections of the vertices, which are shown in Figure~\ref{fi:projections}.  We do not show the edges on the graph, since in the depicted two planes of the 4-dimensional drawing some vertices are positioned in a same point and it could create ambiguities in the representation of the edges. \\
\begin{figure}
	\centering
	\subfigure[]
	{\includegraphics[width=1\linewidth]{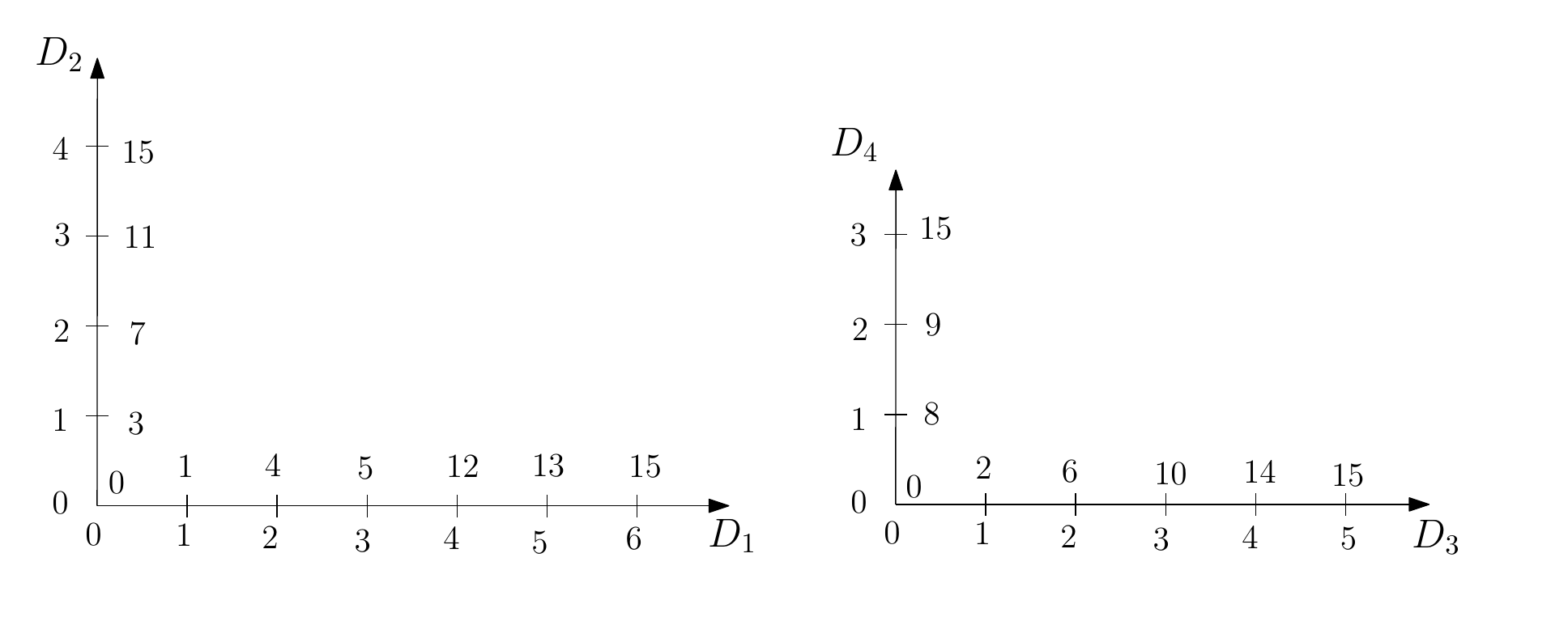}}
	\hfil
	\subfigure[]
	{\includegraphics[width=1.02\linewidth]{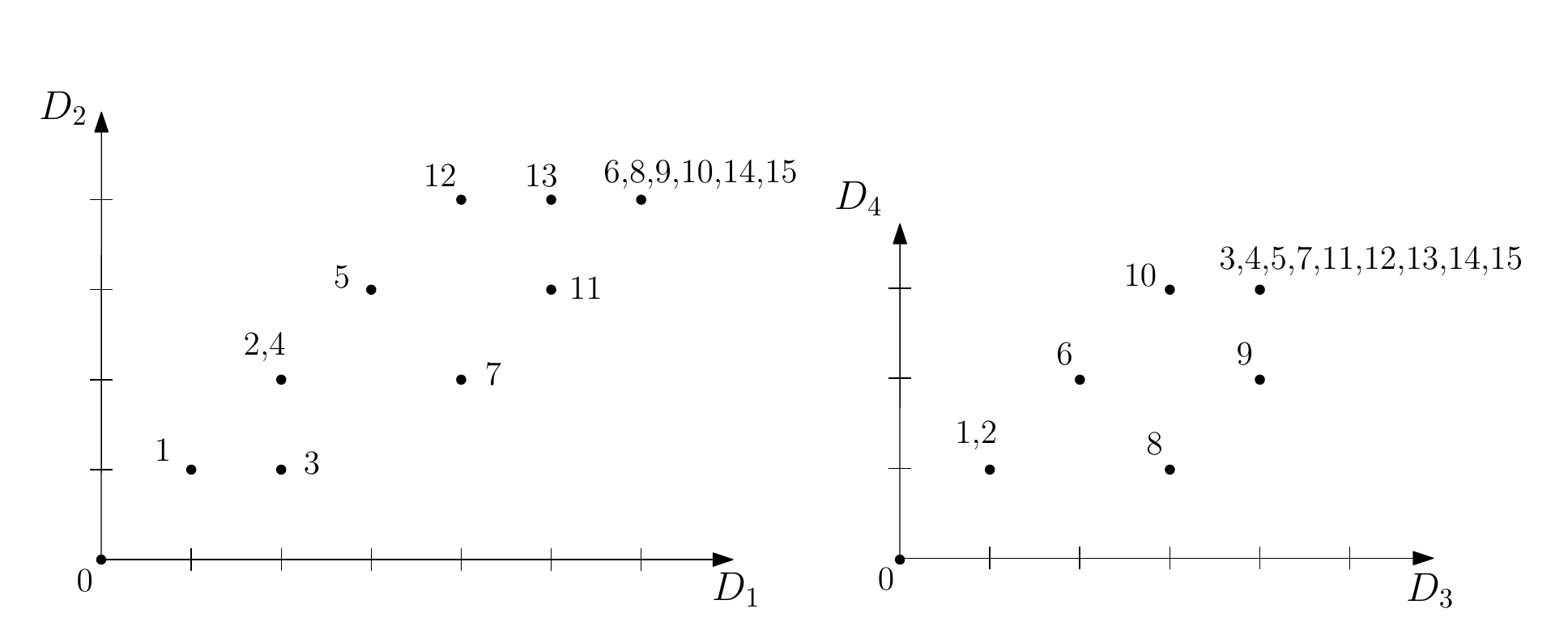}}
	\caption{Illustration of Algorithm kD-Draw on graph $G$ depicted in Figure~\ref{fi:channeldec}(a) and the channel decomposition $S_c$ of $G$ depicted in Figure~\ref{fi:channeldec}(b): (a) shows the initialization step, where the algorithm assigns the value of the coordinate in dimension $D_i$ for every vertex $v\in C_i$ ($0\le i\le 4$);  (b) shows the assignment of the other coordinates, computed by Algorithm kD-Draw in Lines~4-7 by using the projections of the vertices, as shown in Figure~\ref{fi:projections}.  We do not show the edges of the graph, since in the depicted two planes of the 4-dimensional drawing some vertices are positioned in a same point and it could create ambiguities in the representation of the edges.}
	\label{fi:kd}
\end{figure}
\indent
Clearly, if $k=2$, kD-Draw and 2D-Draw compute the same drawing. Hence, kD-Draw is a generalization of 2D-Draw. As we did in the previous section, before proving that $\Gamma$ = kD-Draw($G$, $S_c$) is a dominance drawing we prove some intermediate results.
\begin{lemma}
	\label{lemma:same_position}
	Any two distinct vertices $v,w\in V$ are placed on distinct points in $\Gamma$. 
\end{lemma}
\begin{proof}
	Let $v=(i,j)$ and $u=(h,l)$ be two vertices of $G$. It is sufficient to prove that there exists at least one dimension $D_x$ such that $D_x(v)\not = D_x(u)$ for any pair of vertices $u,v$. Following arguments similar to the ones used in the proof of Lemma~\ref{lemma:same_position2} it is easy to prove that $D_i$ or $D_h$ can be $D_x$.
\end{proof}
 The following two lemmas can be proved using similar arguments to the ones used to prove Lemma~\ref{lemma:source_sink_2} and Lemma~\ref{lemma:reachability_2}:
 \begin{lemma}
 	\label{lemma:source_sink}
 	For any vertex $v\in G$: $s\preceq v$ and $v\preceq t$.
 \end{lemma}
 \begin{lemma}
 	\label{lemma:reachability}
	Let $u\in C_i$ and $v$ be two vertices of $G$. $r(u,v)=yes\Leftrightarrow D_i(u)\le D_i(v)$.
 \end{lemma}
  We are ready now to prove the main result of this section:
\begin{theorem}
	\label{theorem:kDdraw}
	Let $G$ be an st-graph and $S_c=\{C_1,...,C_k\}$ be a channel decomposition of $G$. Given $G$, $S_c$ and the set of all the projections for any vertex (provided by the compressed transitive closure), Algorithm kD-Draw computes a $k$-dimensional dominance drawing $\Gamma$ of $G$. Moreover:
	\begin{itemize}
	    \item[(a)] kD-Draw requires $O(kn)$ time. 
	    \item[(b)] A preprocessing step required to compute all the projections of the vertices of $G$ takes time $O(mk)$.
	    \item[(c)]  If $k$ is required to be equal to the width of $G$, $w_G$, then the preprocessing step requires time $O(n^3)$ or $O(w_Gn^2)$.
	\end{itemize}
 
\end{theorem} 
\begin{proof}
	First, we prove the correctness of the algorithm, i.e., we prove that the drawing computed by our algorithm is a dominance drawing. Then we will prove the time complexity of the algorithm. Let $v=(i,j)$ and $u=(h,l)$ be two different vertices. If $u$ or $v$ are the source or the sink the theorem is clearly true by Lemma~\ref{lemma:source_sink}. Otherwise, without loss of generality, suppose that $i=1$ i.e., $v=(1,j)$. Lemma~\ref{lemma:same_position} proves that $u$ and $v$ are never placed in a same point. Moreover, Lemma~\ref{lemma:reachability} proves that $u \preceq v \Rightarrow r(u,v)=yes$. We need to prove that: $r(u,v)=yes \Leftarrow u\preceq v$. Let $D_p$ be a dimension of $\Gamma$ ($p\in [1,k]$). We prove that $D_p(u)\le D_p(v)$ for any possible value of $p$. First we assume that $p=i$ or $p=h$.  Next, we consider all other cases.
	\begin{enumerate}
		\item  We can prove that $D_p(u)\le D_p(v)$ using arguments similar to the ones used in Case (1) of the proof of Lemma~\ref{lemma:dd2}. 
		\item Suppose $u'=Proj_{C_p} (u)=(p,l')$ and $v'=Proj_{C_p} (v)=(p,j')$. We have $r(u,v')=yes$, since $r(u,v)=yes$ and $r(v,v')$ by definition of projection. The vertex $u'$ has the lowest position in $C_p$ among the vertices of $C_p$ reachable from $u$. Therefore it must be equal to $v'$ or one of its successors. Consequently $l'\le j$ and $D_p(u)\le D_p(v)$. 
	\end{enumerate}
	Finally, regarding the computational time:
	\begin{itemize}
	    \item[(a)] kD-Draw needs $O(n)$ time for the Lines 2-3 and $O(nk)$ time for the Lines 4-7 by Lemma~\ref{lemma:TC}.
	    \item[(b)] We store the projections of the vertices of $G$ in a compressed transitive closure. The time needed to compute the compressed transitive closure is $O(km)$~\cite{DBLP:journals/tods/Jagadish90}.
	    \item[(c)] A channel decomposition with the minimum number of channels, $k$, can be computed in $O(n^3)$ or $O(w_Gn^2)$ time~\cite{DBLP:journals/tods/Jagadish90,chendag}.
	\end{itemize} 
\end{proof}

\paragraph{Dominance Drawing with Distinct Coordinates}  is a dominance drawing where the value of the coordinates of the vertices in every dimension is a topological sorting of the vertices (i.e., distinct coordinates). In the rest of the section we show how we can compute a $k$-dimensional dominance Drawing with distinct coordinates $\Gamma_T$ of $G$ given a  $k$-dimensional dominance drawing $\Gamma$ of $G$.

Suppose that the vertices of $G$ are topologically sorted and let $T(v)$ be the order of vertex $v$ in the topological sorting. Let $D_g$ be a dimension of $\Gamma$ and let $\{v_1,...,v_l\}$ be the set of vertices having the same coordinate in dimension $D_g$, i.e., $D_g(v_1)=D_g(v_2)=...=D_g(v_l)=\alpha$. Suppose $T(v_1)<T(v_2)<...<T(v_l)$. Let $V_{\alpha,g}$ be the set of vertices having coordinate higher than $\alpha$ in dimension $D_g$. We shift the vertices of $V_{\alpha,g}$ by $l$ positions and we shift every vertex $v_i$, $i = 1, 2, ..., l$ by $i-1$ positions. We do continue with this process until there is no pair of vertices placed in the same position. It is easy to see that the resulting drawing $\Gamma_T$ is a dominance drawing with distinct coordinates.

Figure~\ref{fi:topological} shows a dominance drawing with distinct coordinates obtained from the dominance drawings shown in Figure~\ref{fi:kd}(b).
\begin{figure}
	\centering
	{\includegraphics[width=0.7\linewidth]{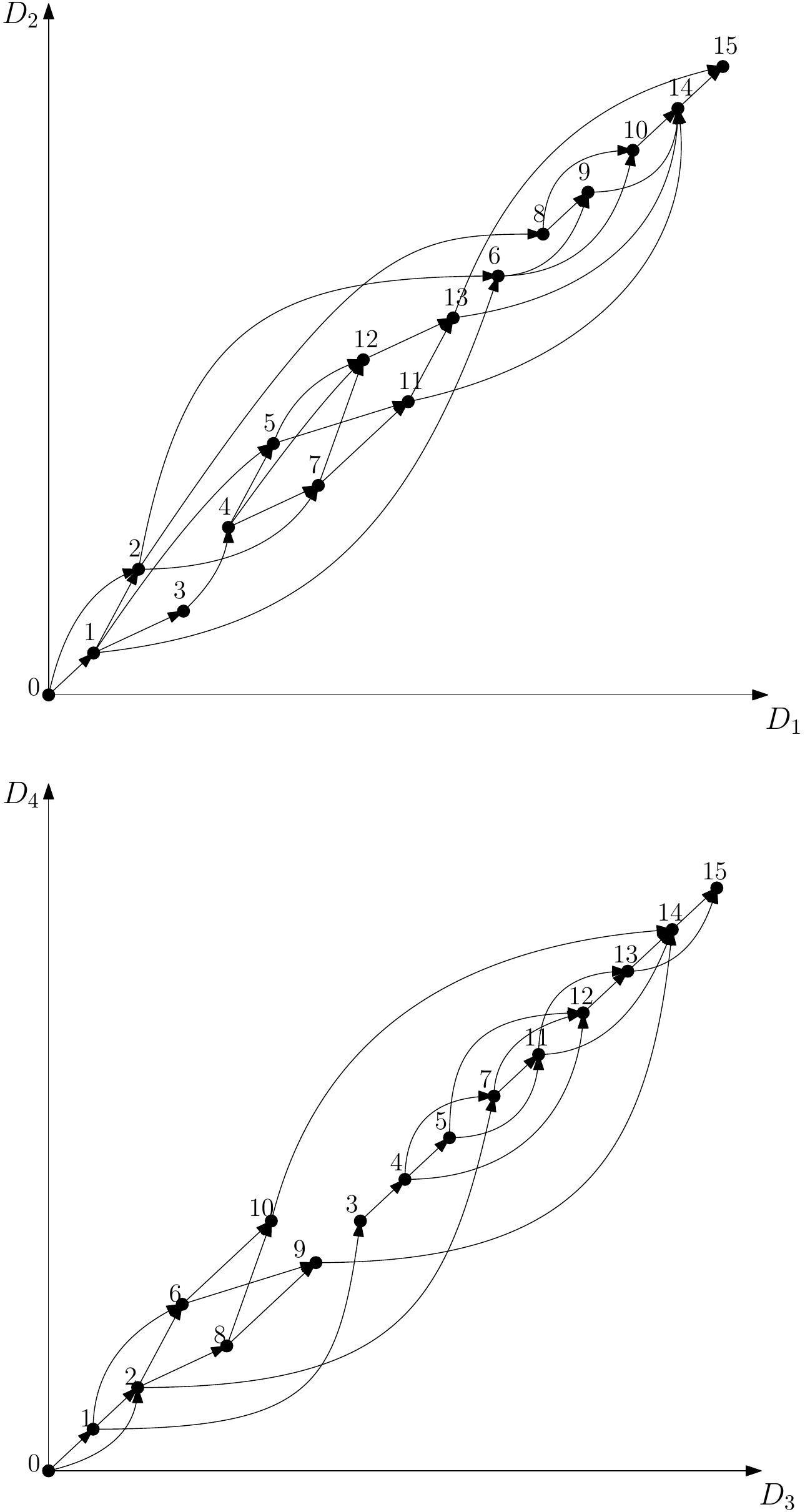}}
	\caption{A dominance drawing with distinct coordinates obtained from the drawings shown in Figure~\ref{fi:kd}(b)}
	\label{fi:topological}
\end{figure}
\subsection{Corollaries}
\label{subsection:corollaries}
The dominance drawing algorithm presented above provides interesting corollaries in the general dimension theory of DAGs. Namely, the following corollary is a direct consequence of  Lemma~\ref{lemma:n/2}, Lemma~\ref{label:width} and Theorem~\ref{theorem:kDdraw}. We point out that Algorithm kD-Draw provides a new upper bound to the dominance dimension of an st-graph:
\begin{corollary}
	\label{corollary:final}
	Let $G$ be any st-graph (or DAG) with $n$ vertices. Then $d_G\le min(\frac{n}{2},w_G)$.
\end{corollary}
As  discussed in the introduction, any st-planar graph has a 2-dimensional dominance drawing~~\cite{DBLP:journals/dcg/BattistaTT92,st-planar}. The next corollary presents a new family of DAGs that have a 2-dimensional dominance drawing.  Obviously these DAGs are not contained in the st-planar family.
\begin{corollary}
	\label{corollary:2D}
	Every DAG $G$ of width 2 has a 2-dimensional dominance drawing.
\end{corollary}
As discussed in the preliminaries, a partial order is a mathematical formalization on the concept of ordering and the results obtained for DAGs and their dominance drawing dimension transfer directly to partial orders and their dimension and vice versa.
Therefore, let $d_P$ be the dimension of partial order $P$. As a consequence of the previous corollary we have the following:
\begin{corollary}
	\label{corollary:partialorder}
	Let $P$ be any partial order with $n$ elements. Then $d_P\le min(\frac{n}{2},w_P)$.
\end{corollary}

\section{Modules and Dominance Drawings}
\label{Section:Modules}
In this section we will exploit the concept of modules in directed acyclic graphs in order to obtain dominance drawings having potentially less dimensions. Let $G=(V,E)$ be an st-graph with $n$ vertices and $m$ edges. A module $M$ of $G$ is a non-empty subset of $V$ such that all vertices in $M$ have the same sets of predecessors and successors in $V-M$. Decomposing a graph into modules may help in various graph problems~\cite{DBLP:journals/dm/McConnellS99}. 

The trivial modules of $G$ are: the set $V$ and the singleton sets \{$v$\}, for any $v\in V$. A graph is called \emph{prime} if it does not possess non-trivial modules. Two modules \emph{overlap} if they have a non empty intersection, and one does not contain the other. A module is a \emph{strong module} if it does not overlap with any module, otherwise it is a \emph{weak module}.

A modular decomposition of $G$ is a representation of all modules of $G$. The decomposition forms a tree, whose nodes are the strong modules of $G$, ordered by the subset relationship. In particular, the root of the modular decomposition tree is $V$ , and its leaves are the singleton sets.

The \emph{congruence partition} $C_P=\{M_1,...,M_h\}$ of $V$ is a partition of $V$ into modules (i.e., each vertex belongs to exactly one module $M_i$, with $1 \le i \le h$). It is easy to obtain a congruence partition from a modular decomposition.  Each DAG has an exponential number of congruence partitions~\cite{DBLP:journals/dam/McConnellM05,DBLP:journals/dm/McConnellS99}. The \emph{quotient graph} $G_0$ of $G$ given $C_P$ is the graph obtained from $G$ by merging the nodes of each module in $C_P$.  We denote by $\mu_i$ the vertex representing $M_i$ in $G_0$. The following lemma is true due to the definitions of module and of quotient graph:
\begin{lemma}\label{lemma:property_quotient_graph}
	Let $u,v$ be two vertices of $G$ such that $u\in G_i$ and $v\in G_j$. Vertex $\mu_j$ is reachable from vertex $\mu_i$ in $G_0$ if and only if $v$ is reachable from $u$ in $G$.
\end{lemma}

Let $G^*$ be the transitive closure of $G$. A module $M$ of $G^*$ is a \emph{transitive module} of $G$. The \emph{transitive congruence partition} $C_P$ of $V$ is a partition of $V$ into transitive modules. The \emph{transitive quotient graph} $G/C_P$ is the graph obtained
from $G$ by merging the nodes of each transitive module in $C_P$.

Figure~\ref{fi:congruenceandquotient}(a) shows a transitive congruence partition $C_P=\{M_1,M_2,M_3,M_4\}$ of the graph $G$ depicted in Figure~\ref{fi:channeldec}(a), where: $M_1=\{$ $0,$ $1,$ $2\}$; $M_2=\{3,$ $4,$ $5,$ $7,$ $11,$ $12,$ $13\}$; $M_3=\{6,8,9,10\}$; $M_4=\{14,15\}$. Figure~\ref{fi:congruenceandquotient}(b) shows the quotient graph $G_0$ of $G$ given $C_P$.
\begin{figure}
	\centering
	\subfigure[]
	{\includegraphics[width=0.4\linewidth]{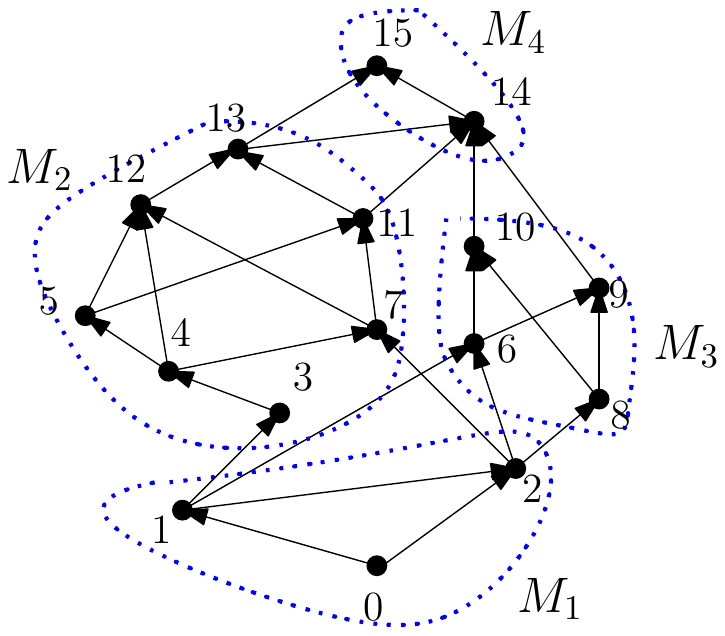}}
	\hfil
	\subfigure[]
	{\includegraphics[width=0.4\linewidth]{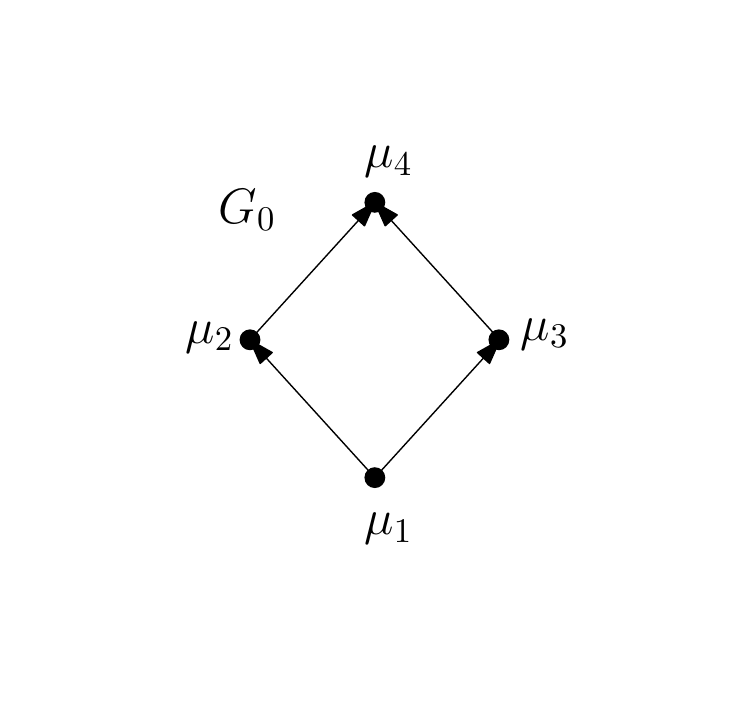}}
	\caption{(a) A transitive congruence partition $C_P=\{M_1,M_2,M_3,M_4\}$ of graph $G$ shown in Figure~\ref{fi:channeldec}(a), where: $M_1=\{0,1,2\}$; $M_2=\{3,4,5,7,11,12,13\}$; $M_3=\{6,8,9,10\}$; $M_4=\{14,15\}$;  (b) the quotient graph $G_0$.}
	\label{fi:congruenceandquotient}
\end{figure}
\noindent
\\\\
A modular decomposition of $G$ can be computed in linear time~\cite{DBLP:journals/dam/McConnellM05,DBLP:journals/dm/McConnellS99}. Hence, computing a transitive modular decomposition of $G$ requires $O(n m)$ times, since it is equivalent to computing the modular decomposition of $G^*$. From now on we will only deal with transitive modules, transitive congruence partitions and transitive quotient graph, so in our description we will omit the term "transitive".

Let $C_P=\{M_1,...,M_h\}$ be a congruence partition of $G$ and let $G_0$ be the quotient graph $G/C_P$. We denote by $\mu_i$ the vertex representing $M_i$ in $G_0$.  We denote by \emph{module-induced graph} of $M_i$ the graph $G_i=(M_i,E(M_i))$, where $E(M_i)$ is the subset of edges of $E$ which are incident to two vertices of $M_i$. Without loss of generality we assume that every $G_i$ is an st-graph. If $G_i$ is not an st-graph we do the operations described in the next paragraph.

If $G_i$ contains the source of the graph, we add a virtual sink $t_i$ to it; if it contains the sink of the graph, we add a virtual source $s_i$ to it; else, we add a virtual source $s_i$ and a virtual sink $t_i$ to it. Then: we add some edges connecting all the vertices of $V-M_i$ reaching $M_i$ to $s_i$; we add some edges connecting $t_i$ to all the vertices of $V-M_i$ reached by $M_i$; we remove all the edges of $V-M_i$ adjacent to a vertex of $M_i$ different from $s_i$ and $t_i$. 

Figure~\ref{fi:module_st} shows the graph obtained by adding a sink and a source to every module of the graph depicted in Figure~\ref{fi:congruenceandquotient} following the steps described in the above paragraph.
\begin{figure}
	\centering
	{\includegraphics[width=0.7\linewidth]{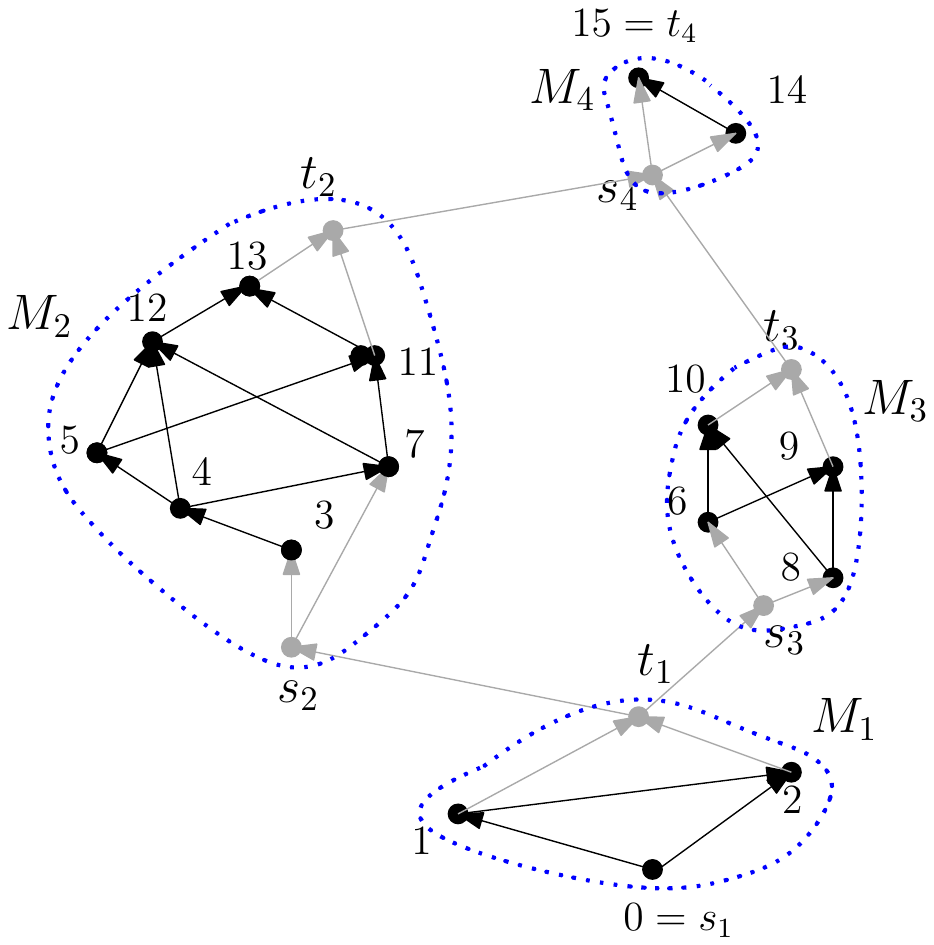}}
	\caption{The graph obtained from the graph $G$ shown in Figure~\ref{fi:congruenceandquotient}  after adding a sink and a source to every module of the $G$.}
	\label{fi:module_st}
\end{figure}
\\
\noindent
Figure~\ref{fi:induced} shows the induced graphs of graph $G$ shown in Figure~\ref{fi:module_st}. 
\begin{figure}
	\centering
	{\includegraphics[width=1\linewidth]{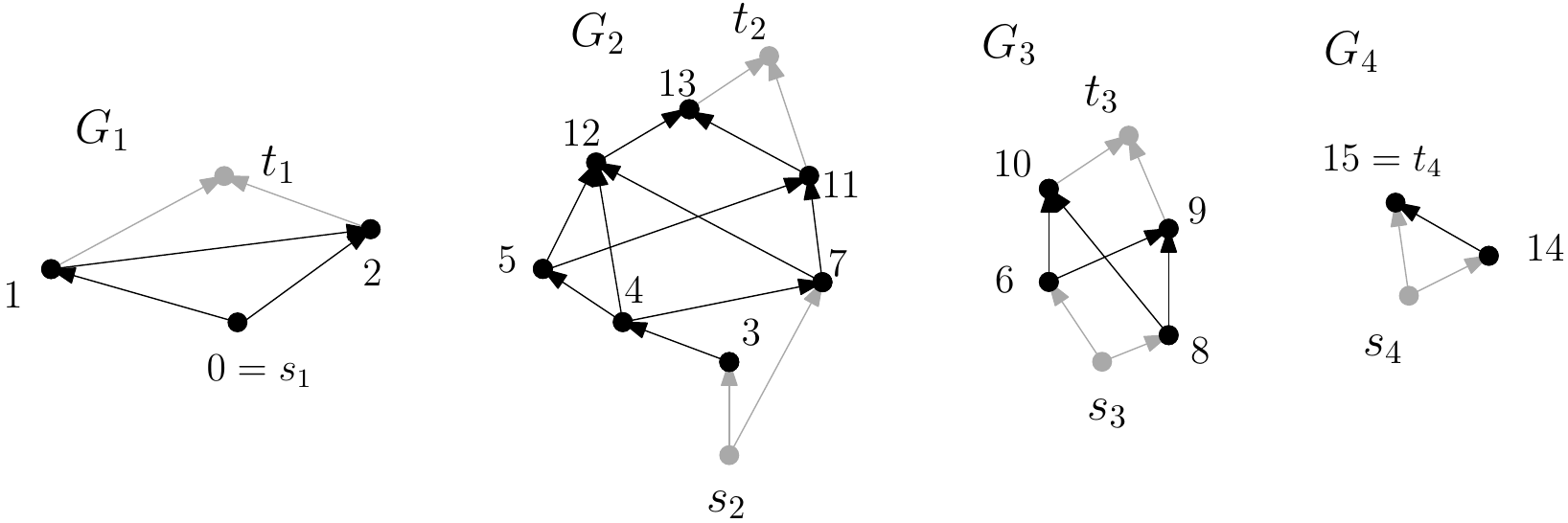}}
	\caption{The induced graphs of graph $G$ shown in Figure~\ref{fi:module_st}.}
	\label{fi:induced}
\end{figure}

Let $G(C_P)=\{G_0,G_1,...,G_h\}$ be the set of all the module-induced graphs of $G$ augmented with $G_0$. We denote by $w_i$ the width of $G_i$ (i.e., $w_{G_i}=w_i$). We denote by \emph{dimensional neck} $w_N$ of $C_P$ and $G(C_P)$ the value of the maximum width among all the graphs in $G(C_P)$.  Any graph $G_i\in G(C_P)$ is isomorphic to a subgraph of $G^*$ and the width of $G$ is equal to the width of it
s transitive closure. Hence, $w_N\le w_{G^*}=w_{G}$.

Let $G(C_P)=\{G_0,G_1,...,G_h\}$ be the set of all the module-induced graphs of $G$ augmented with $G_0$. We denote by $w_i$ the width of $G_i$ ($w_{G_i}=w_i$). We denote by \emph{dimensional neck} $w_{N}$ of $C_P$ and $G(C_P)$ the value of the maximum width among all the graphs in $G(C_P)$.  Any graph $G_i\in G(C_P)$ is isomorphic to a subgraph of $G^*$ and the width of $G$ is equal to the width of its transitive closure. This means that $w_{n}\le w_{G^*}=w_{G}$. Hence we have the following:

\begin{lemma} 
	The dimensional neck of $G(C_P)$ is not greater than the width of $G$. In other words, $w_{N}\le w_G$.
\end{lemma}

We are interested in calculating a dominance drawing $\Gamma_i$ for each graph $G_i$ of $C_P$. Moreover, we want to merge these drawings into a single dominance drawing of $G$. In order to do that, we want the drawings $\Gamma_0,...,\Gamma_h$ to have the same number of dimensions. Notice that, given a channel decomposition of a graph $G_i$ having size $\alpha$, it is always possible to compute a channel decomposition of $G_i=(V_i,E_i)$ having size $\alpha+\beta$, for any $\beta$, by adding to the channel decomposition $\beta$ (dummy) channels containing only the source and the sink of $G_i$.

The following algorithm, called Algorithm Drawings-Computation, receives as input the set $G(C_P)=\{G_0,...,G_h\}$ and it gives as output the set of $w_{N}$-dimensional dominance drawings $\Gamma_1,...,\Gamma_h$, where $\Gamma_i$ is a $w_{N}$-dimensional dominance drawing of $G_i$, for any $i= 1, 2, ..., h$. 

In Line~1 we compute the dimensional neck $w_N$ of $G(C_P)$. This is trivial, since we can simply compute the width of any DAG in $G(C_P)$ (as described in~\cite{DBLP:journals/tods/Jagadish90}) and pick up the highest one. This operation requires linear time. In Lines~2-3 we compute a channel decomposition $S_C^i$ of minimum size for any $G_i$ by using Algorithm Channels-Generation, which was introduced in Section~\ref{Section:projections}. In Lines~4-6 we add channels containing only the source and the sink of $G_i$ to any $S_C^i$ until its size is $w_N$. Then we compute a $w_N$-dimensional dominance drawing $\Gamma_i$ of each $G_i$ using Algorithm kD-Draw, as described in Section~\ref{Section:Multidimensional}.
\\\\\\\\
\textbf{Algorithm} Drawings-Computation($G(C_P)$)\\
1. \indent Compute the dimensional neck $w_N$ of $G(C_P)$.\\
2. \indent \textbf{For} each $G_i=G_0,G_1,...,G_h$\\
3. \indent \indent $S_C^i=$ Channels-Generation$(G_i)$\\
4. \indent \indent \textbf{While} $S_c^i$ contains less than $w_N$ channels\\
5. \indent \indent \indent $C_{|S_c^i|+1}=\{s_i,t_i\}$\\
6. \indent \indent \indent $S_C.add(C_{|S_c^i|+1})$\\
7. \indent \indent $\Gamma_i=$ kD-Draw($G_i$, $S_C^i$)\\
8. \indent \emph{output:} $\Gamma_0,...,\Gamma_h$\\\\
\indent
Figure~\ref{fi:drawing_computation} shows the output of Algorithm Drawings-Computation given the set $G(C_P)=\{G_0,G_1,G_2,G_3,G_4\}$, where $G_0$ is depicted in Figure~\ref{fi:congruenceandquotient}(b) and the induced graphs $G_1$, $G_2$, $G_3$, and $G_4$ are depicted in Figure~\ref{fi:induced}. Part~(a) shows: a channel decomposition $\{C_{01},C_{02}\}$ of minimum size of $G_0$, computed in Line~3 of Algorithm Drawings-Computation, where $C_{01}=\{\mu_1,\mu_2,\mu_4 \}$ and $C_{02}=\{ \mu_1,\mu_3,\mu_4 \}$; the dominance drawing $\Gamma_0$ of $G_0$, computed in Line~7 of Algorithm Drawings-Computation. Similarly: Part~(b) shows $\{C_{11},C_{12}\}$, where $C_{11}=\{0,1,t_1\}$ and $C_{12}=\{0,2,t_1\}$, and $\Gamma_1$; Part~(c) shows $\{C_{21},C_{22}\}$, where $C_{21}=\{s_2,3,4,5,12,13,t_2\}$ and $C_{22}=\{s_2,7,11,t_2\}$, and $\Gamma_2$; Part~(d) shows $\{C_{31},C_{32}\}$, where $C_{31}=\{s_3,6,10,t_3\}$ and $C_{32}=\{s_3,8,9,t_3\}$, and $\Gamma_3$; Part~(e) shows $\{C_{41},C_{42}\}$, where $C_{41}$ is the empty channel $\{s_4,15\}$ and $C_{42}=\{s_4,14,15\}$, and $\Gamma_4$.\\
\indent
\begin{figure}
	\centering
	\subfigure[]
	{\includegraphics[width=0.45\linewidth]{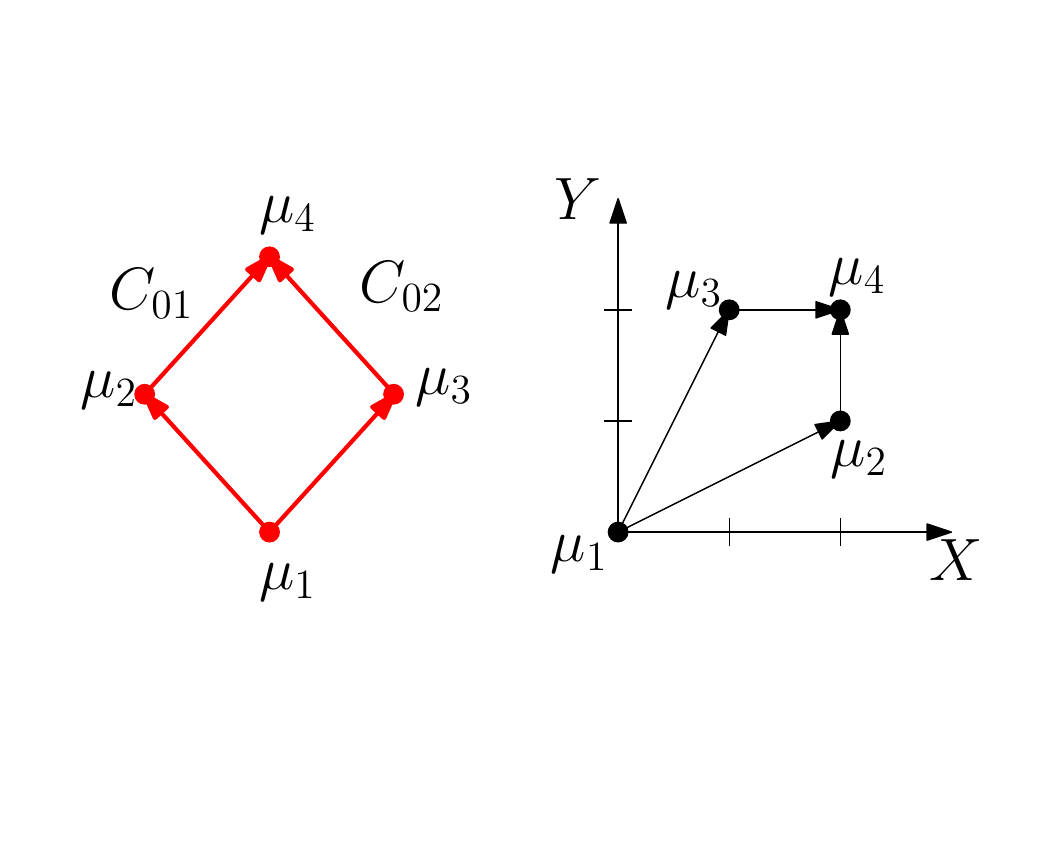}}
	\hfil
	\subfigure[]
	{\includegraphics[width=0.45\linewidth]{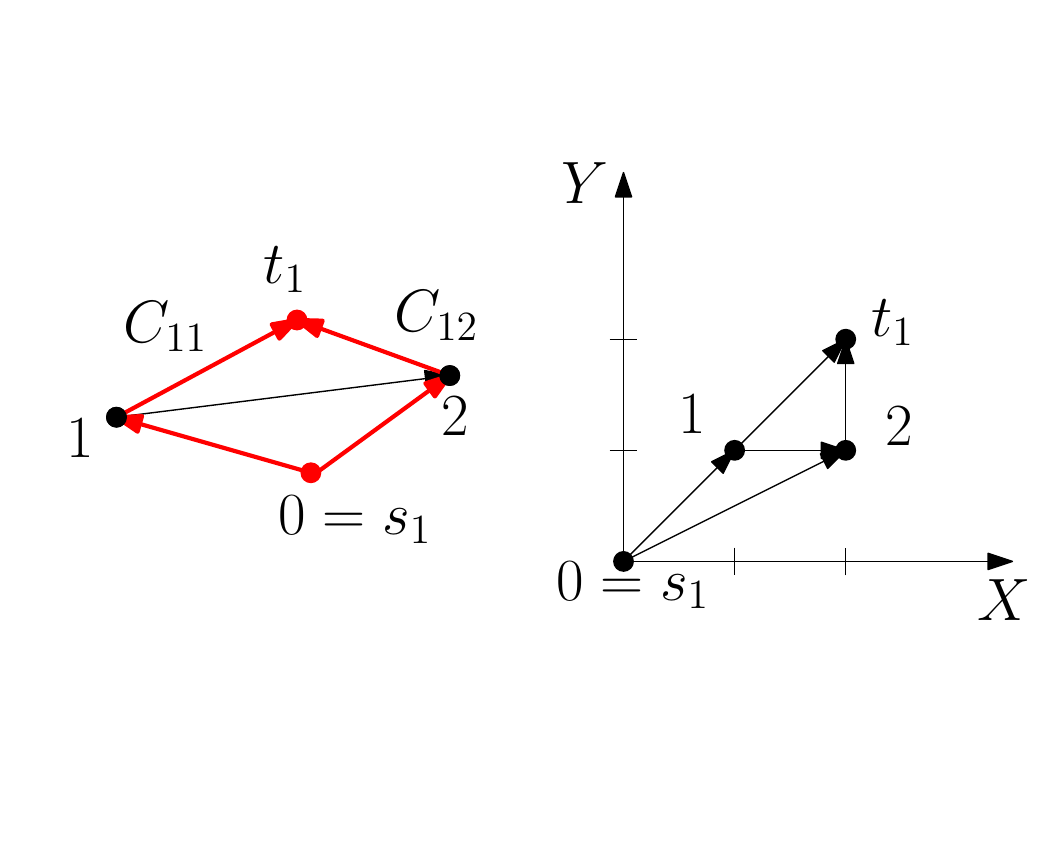}}
	\hfil
	\subfigure[]
	{\includegraphics[width=0.45\linewidth]{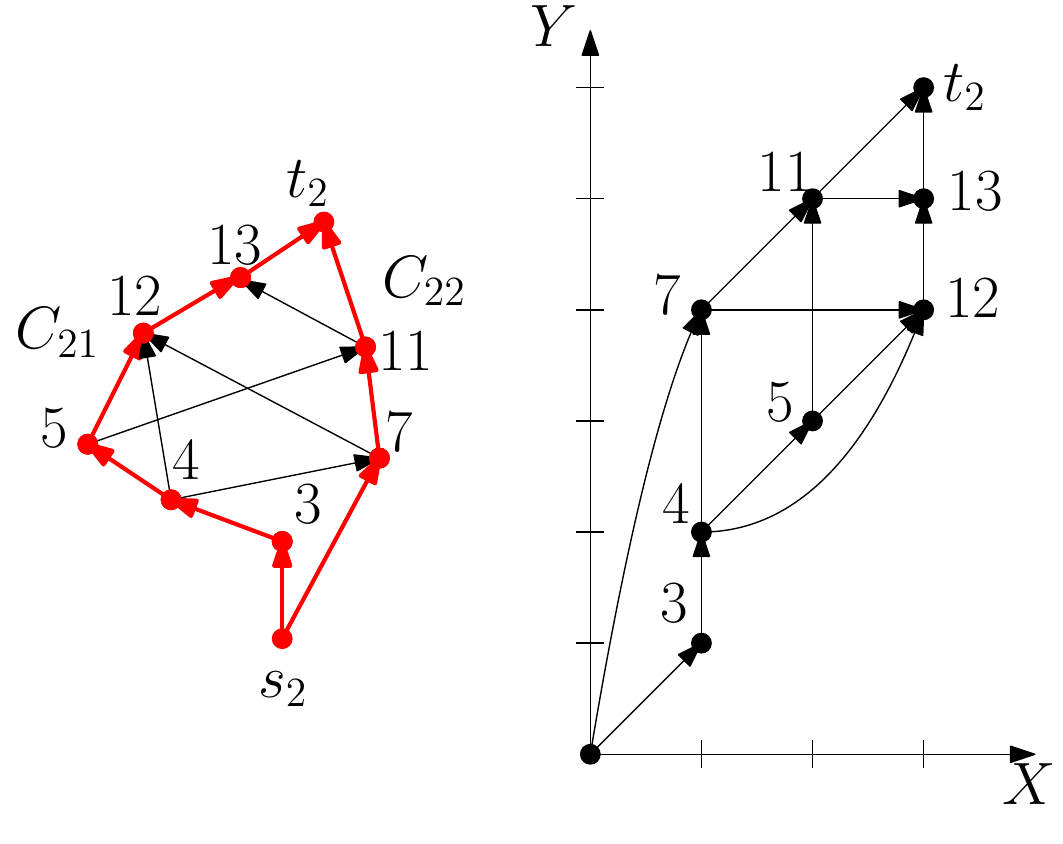}}
	\hfil
	\subfigure[]
	{\includegraphics[width=0.45\linewidth]{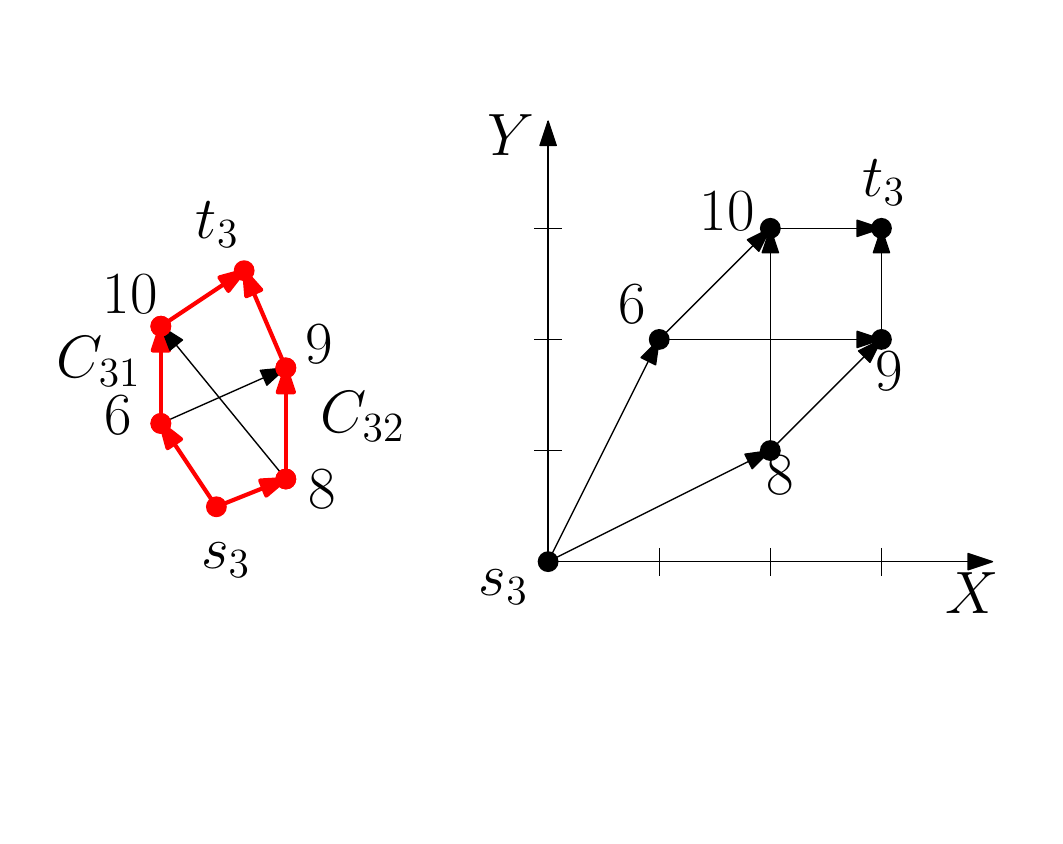}}
	\hfil
	\subfigure[]
	{\includegraphics[width=0.45\linewidth]{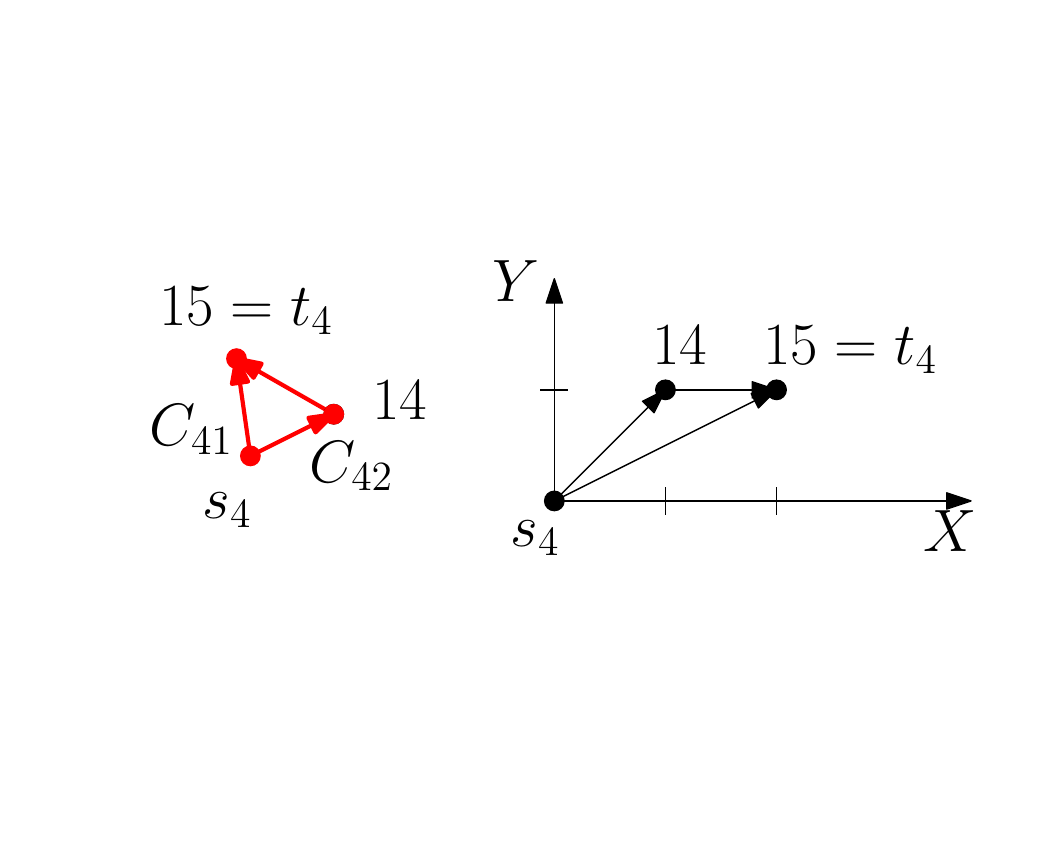}}
	\caption{The output of Algorithm Drawings-Computation given the set $G(C_P)=\{G_0,G_1,G_2,G_3,G_4\}$, where $G_0$ is depicted in Figure~\ref{fi:congruenceandquotient}(b) and the induced graphs $G_1$, $G_2$, $G_3$, and $G_4$ are depicted in Figure~\ref{fi:induced}.}
	\label{fi:drawing_computation}
\end{figure}
We recall that computing $w_N$ from a graph requires linear time~\cite{DBLP:journals/tods/Jagadish90}. It is easy to see that the time complexity of the algorithm depends on Line 3, where we compute the minimum-size channel decomposition of any graph in $G(C_P)$, and Line 7, where we compute the dominance drawings $\Gamma_0,...,\Gamma_h$. In order to characterize this time complexity we give some further definitions. We denote by $n_i$ the number of vertices of $G_i$. We recall that the number of vertices of $G_0$ is equal to the number of modules of $C_P$ and that the number of vertices of any $G_i$ is equal to the number of vertices of the module $M_i$. Hence, $n_0=|C_P|$ and $n_i=|M_i|$ for any $0<i\le h$. Let $\rho$ be the maximum $n_i$ for $i\in[0,h]$ and let $w_{\rho}$ be the maximum width among the graphs of $G(C_P)$ having $w_{\rho}$ vertices.
\begin{lemma}\label{lemma:etadneck}
	The maximum number of vertices of a module of $C_P$, $\rho$, is greater than or equal to the dimensional neck of $G(C_P)$. In other words, $w_N\le \rho$.
\end{lemma}
\begin{proof}
	Let $G_i$ be the graph such that $w_i=w_N$. We have that $w_N$ is less than $n_i$ by Lemma~\ref{lemma:n/2} and $n_i$ is less than $\rho$ by the definition of $\rho$. In other words: $w_N\le n_i\le \rho$.
\end{proof}

The time complexity of Line 3 is $O(w_{\rho} \rho^2)$, Lemma~\ref{lemma:channeldecomposition}, and the time complexity of Line 8 is $O(w_N \rho)$. Moreover, $O(w_N \rho)\in O(w_{\rho} \rho^2)$ by Lemma~\ref{lemma:etadneck}. It means that the computational complexity of Algorithm Drawings-Computation depends on Line 3. Hence, we have the following lemma:

\begin{lemma}\label{lemma:timecomplexity_drawings_computation}
	Algorithm Drawings-Computation requires $O(w_{\rho} \rho^2)$ time.
\end{lemma}
\indent
Lemma~\ref{lemma:timecomplexity_drawings_computation} tells that the time complexity of the algorithm depends on the maximum number of vertices $\rho$ belonging to a same module and on the width $w_{rho}$ of the graph having such a number of vertices.

As discussed earlier, our final goal is to merge the $w_N$-dimensional dominance drawings $\Gamma_0,...,\Gamma_h$ of graphs $G_0,...,G_h$ into a $w_N$-dominance drawing $\Gamma$ of the original graph $G$. Our strategy is to create space in the dominance drawing $\Gamma_0$ of the quotient graph $G_0$ in order to insert the drawings $\Gamma_1,...,\Gamma_h$ in it by simply using vertex $\mu_i$ as the new origin of drawing $\Gamma_i$ for any $i\in [1,h]$.

We create space by shifting, as described in Algorithm Shifter. It receives as input the drawings $\Gamma_0,...,\Gamma_h$ and gives as output a modified drawing $\Gamma_0$, where the vertices are shifted. Suppose that we want to shift $\mu_j$ with respect of $\mu_i$ in dimension $D_g$ (notice that all the possible ordered couples of vertices of $G_0$ are chosen in Lines 1-2 and that all the dimensions are chosen for every ordered couple of vertices in Line 3).  We shift $\mu_j$ if and only if one of the following conditions is true: its coordinate in some dimension $D_g$ is greater than the corresponding coordinate of $\mu_i$; the coordinates of $\mu_i$ and $\mu_j$ are equal (in dimension $D_g$) and $\mu_j$ is reachable from $\mu_i$. This check is done in Line 4. In that case, let $y$ be the maximum value of a coordinate in $D_g$ in drawing $\Gamma_j$, i.e., the coordinate of $t_j$ (Line 6). We shift $\mu_j$ in dimension $D_g$ by $y$ positions (Line 7).\\\\
\noindent
\textbf{Algorithm} Shifter($\Gamma_0,...,\Gamma_h$)\\
1. \indent \textbf{For} each $\mu_i\in G_0$:\\ 
2. \indent \indent \textbf{For} each $\mu_j\in G_0$ such that $i\not=j$:\\ 
3. \indent \indent \indent \textbf{For} each $D_g=D_1,...,D_{w_N}$\\
4. \indent \indent \indent \indent \textbf{If} $[D_g(\mu_i)<D_g(\mu_j)\lor (D_g(\mu_i)=D_g(\mu_j)\land r(\mu_i,\mu_j))]$\\
5.\indent \indent \indent \indent \indent Let $y$ be equal to the coordinate of $t_j$ in $\Gamma_j$\\
6.\indent \indent \indent \indent \indent $D_g(\mu_j)+=y$\\
7.\indent \emph{output:} $\Gamma_0$\\\\
Figure~\ref{fi:shifting} is an illustration of Algorithm Shifter. The algorithm receives as input the drawings depicted in Figure~\ref{fi:drawing_computation}. Part~(a) shows drawing $\Gamma_0$. Part~(b) shows the drawing $\Gamma_0$ after we perform the shifting to accommodate $\Gamma_1$. Part~(c) shows drawing $\Gamma_0$ after we perform the shifting to accommodate $\Gamma_2$. Finally, Part~(d) shows drawing $\Gamma_0$ after we perform the shifting to accommodate $\Gamma_3$. The space reserved for each $\Gamma_1$, $\Gamma_2$, and $\Gamma_3$ during the algorithm is shown with a red box.
\begin{figure}
	\centering
	\subfigure[]
	{\includegraphics[width=0.165\linewidth]{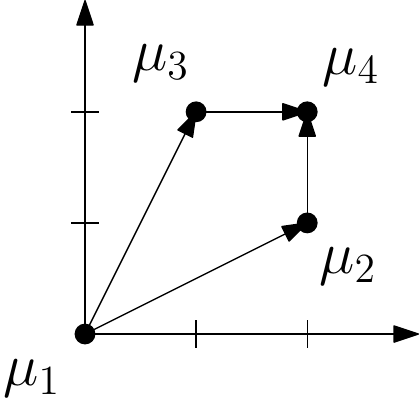}}
	\hfil
	\subfigure[]
	{\includegraphics[width=0.265\linewidth]{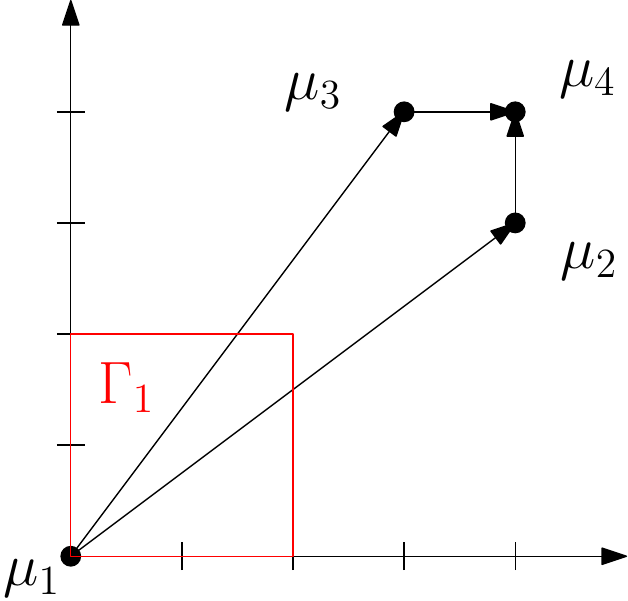}}
	\hfil
	\subfigure[]
	{\includegraphics[width=0.465\linewidth]{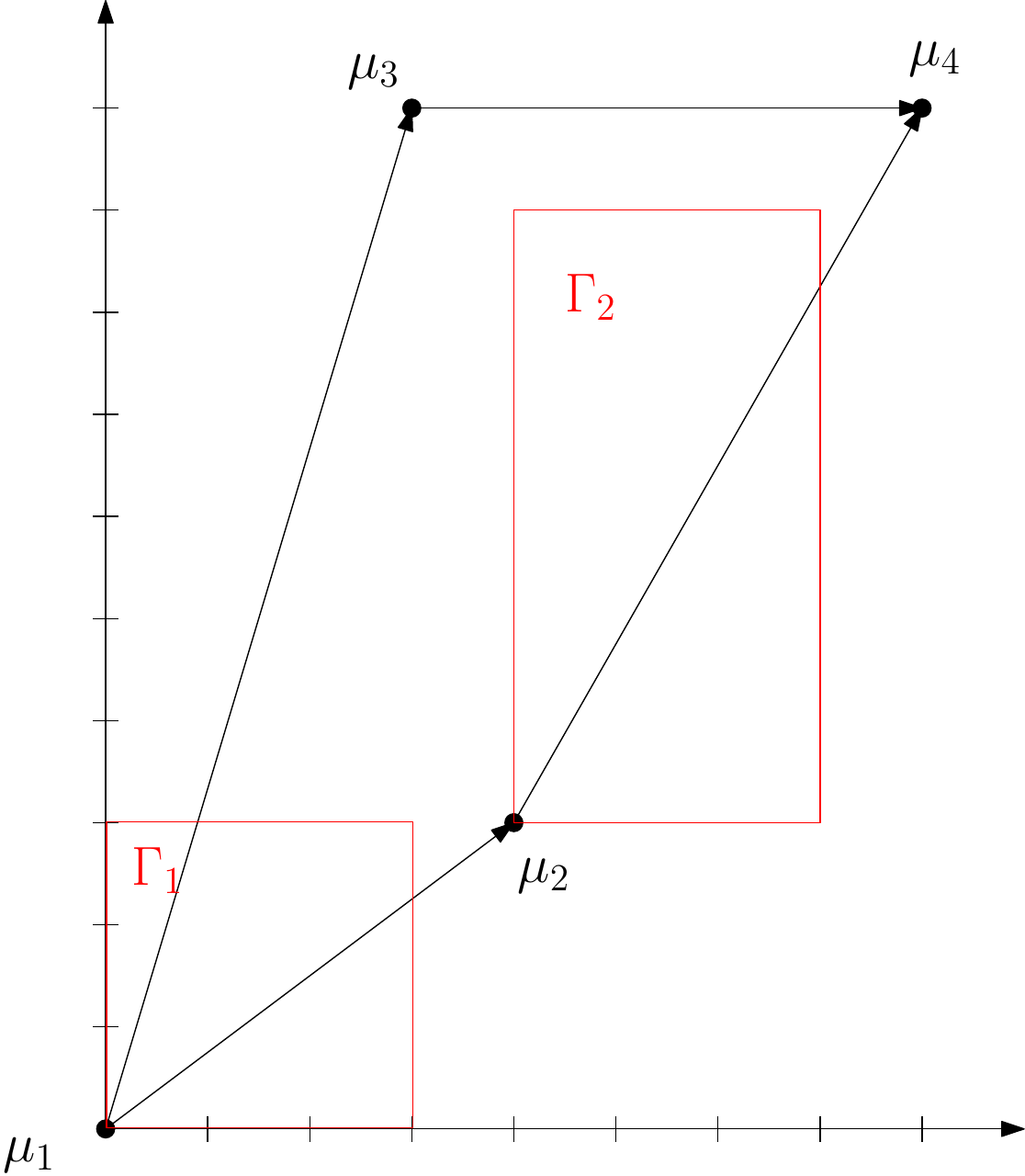}}
	\hfil
	\subfigure[]
	{\includegraphics[width=0.75\linewidth]{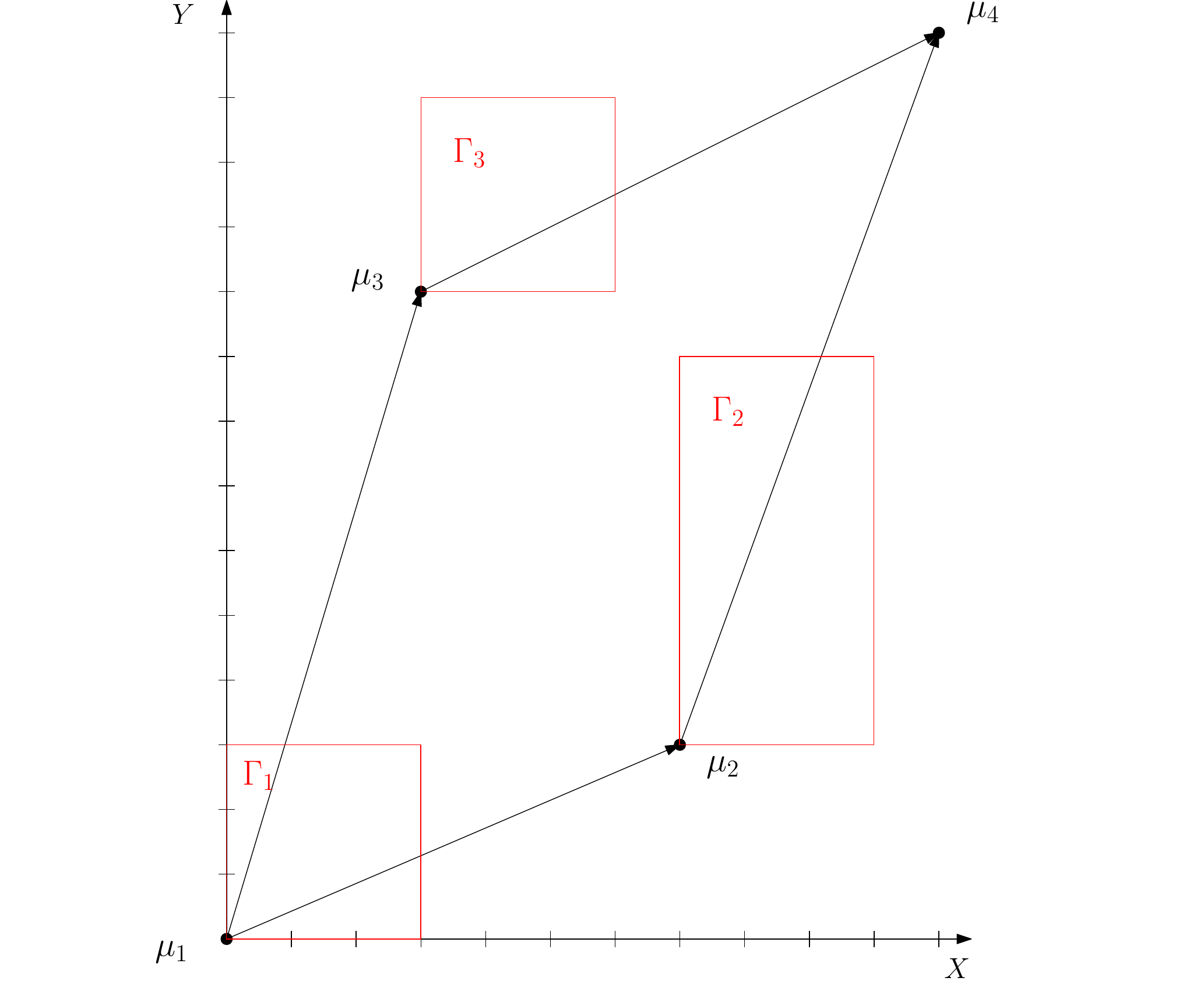}}
	\caption{Illustration of Algorithm Shifter using as input the drawings depicted in Figure~\ref{fi:drawing_computation}: (a) shows the drawing $\Gamma_0$;  (b) shows $\Gamma_0$ after the shifting done with respect to $\Gamma_1$;  (c) shows $\Gamma_0$ after the shifting done with respect to $\Gamma_2$;  (d) shows $\Gamma_0$ after the shifting done with respect to $\Gamma_3$. The space reserved for $\Gamma_1$, $\Gamma_2$, and $\Gamma_3$ during the algorithm is always shown with red boxes.}
	\label{fi:shifting}
\end{figure}
\indent
\\
Recall that the number of vertices of $G_0$ is equal to $h$, which is the number of modules of $C_P$. Moreover, the number of dimensions of $\Gamma_0$ are $w_N$. Algorithm Shifter shifts in constant time every vertex of $G_0$ with respect to every other vertex of $G_0$ in every dimension of $\Gamma_0$ in constant time. Hence, we have the following lemma: 
\begin{lemma}\label{lemma:shiftingtime}
	The time complexity of Algorithm Shifter is $O(w_Nh^2)$. 
\end{lemma}

Finally, putting everything together, we present Algorithm Neck-Dimensional-Draw (or simply ND-Draw). It takes as input $G(C_P)$ and produces as output a dominance drawing of $G$.  First we compute the dominance drawings $\Gamma_0,...,\Gamma_h$ of the module induced graphs $G_0,...,G_h\in G(C_P)$ by using Algorithm Drawings-Computation (Line~1 of ND-Draw). Next we use Algorithm Shifter to shift the vertices of $\Gamma_0$ in order to create space for the vertices of $G_1,...,G_h$, which are already placed in $\Gamma_1,...,\Gamma_h$ (Line~2). Finally, we place every vertex of $v\in G_i$ in $\Gamma_0$ (Lines~3-8). We denote by $D_g(v,\Gamma_0)$ the coordinate of the vertex $v$ in the dimension $D_g$ of drawing $\Gamma_0$. The coordinate in the dimension $D_g$ of $v$ in $\Gamma_0$ is the sum of the coordinate of $\mu_i$ in $\Gamma_0$ with the coordinate of $v$ in $\Gamma_i$ (Line~6). This calculation is made for every $v\in V$. The output of the algorithm is a dominance drawing $\Gamma$ of $G$ having $w_N$ dimensions. \\\\\\\
\textbf{Algorithm} ND-Draw($G(C_P)$)\\
1.\indent $\Gamma_0,...,\Gamma_h=$ Drawings-Computation($G(C_P)$)\\
2.\indent $\Gamma_0=$ Shifter($\Gamma_0,...,\Gamma_h$)\\
3.\indent \textbf{For} each $\Gamma_i=\Gamma_1,...,\Gamma_h$:\\ 
4.\indent \indent \textbf{For} any $v\in G_i$\\
5.\indent \indent \indent \textbf{For} each $D_g=D_1,...,D_{w_N}$\\
6.\indent \indent \indent \indent $D_g(v,\Gamma_0)=D_g(\mu_i,\Gamma_0)+D_g(v,\Gamma_i)$\\
7.\indent \indent \indent \indent remove $\mu_i$ from $\Gamma_0$\\
8.\indent \emph{output:} $\Gamma=\Gamma_0$
\\\\
Figure~\ref{fi:final_1} shows a representation of the output of Algorithm ND-Draw when the input of the algorithm is the set $G(C_P)=\{G_0,G_1,G_2,G_3,G_4\}$, where $G_0$ is depicted in Figure~\ref{fi:congruenceandquotient}(b) and the induced graphs $G_1$, $G_2$, $G_3$, and $G_4$ are depicted in Figure~\ref{fi:induced}. This drawing is obtained by substituting the drawings in Figure~\ref{fi:drawing_computation} in the respective space reserved for them in Figure~\ref{fi:shifting} as described by Algorithm ND-Draw.
\begin{figure}
	\centering
	{\includegraphics[width=1\linewidth]{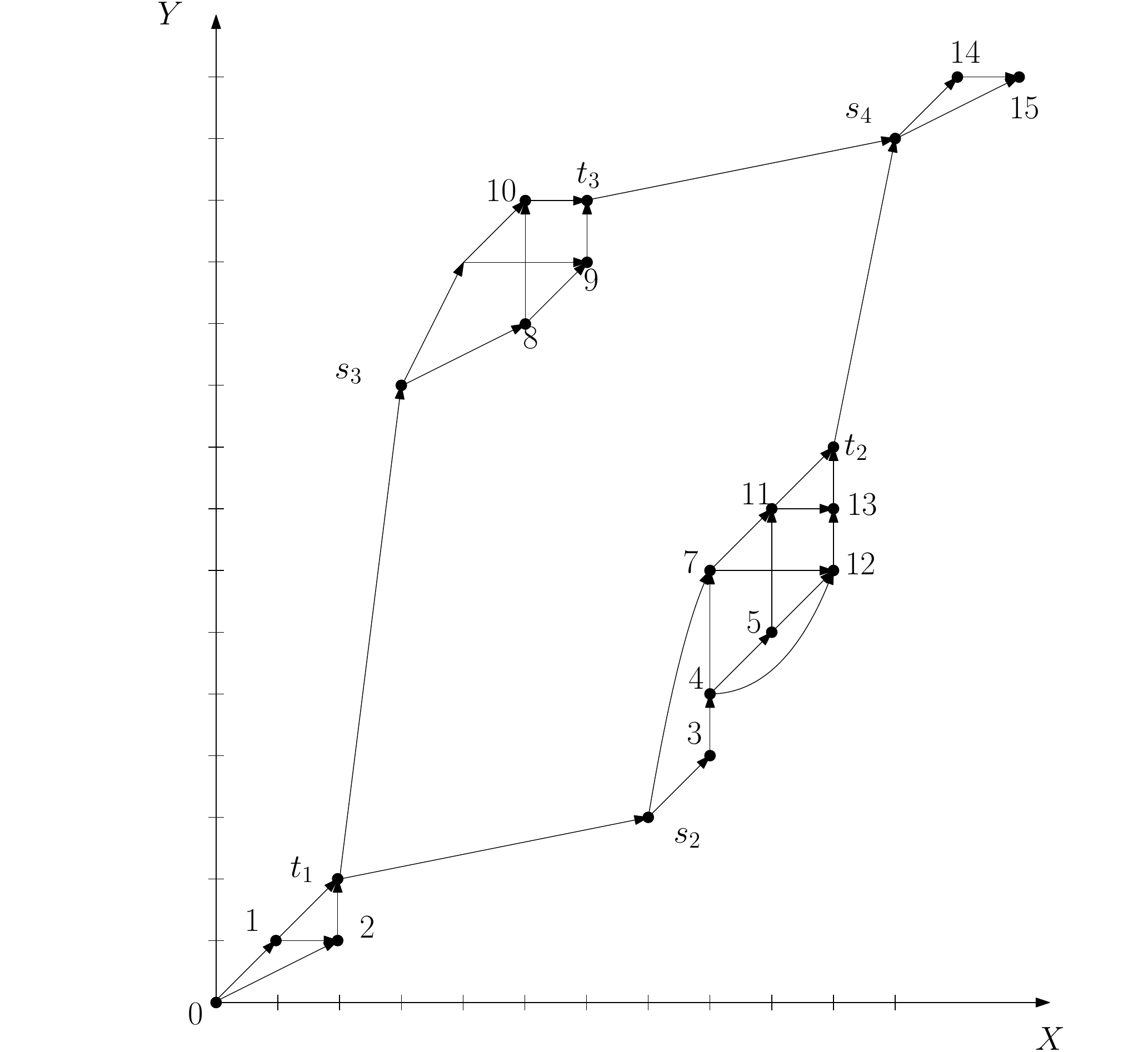}}
	\caption{Output of Algorithm ND-Draw when the input is the set $G(C_P)=\{G_0,G_1,G_2,G_3,G_4\}$, where $G_0$ is depicted in Figure~\ref{fi:congruenceandquotient}(b) and the induced graphs $G_1$, $G_2$, $G_3$, and $G_4$ are depicted in Figure~\ref{fi:induced}. This drawing is obtained by substituting the drawings in Figure~\ref{fi:drawing_computation} in the respective space reserved for them of Figure~\ref{fi:shifting} as described by Algorithm ND-Draw.}
	\label{fi:final_1}
\end{figure}

Figure~\ref{fi:final_2} shows the drawing of Figure~\ref{fi:final_1} where the added sources and sinks are removed and the original edges of the graph are restored.  Notice that this is the same DAG as the one shown in Figure 1 along with its four channels.  According to that channel decomposition, the DAG requires four dimensions, whereas, after using the module decomposition technique it requires only two dimensions.

\paragraph{Important Observation on the New Dominance Drawing Technique}  The graph shown in Figure~\ref{fi:final_1} and Figure~\ref{fi:final_2} is the same graph $G$ shown in Figure~\ref{fi:channeldec}.  Recall that $G$ has a minimum size channel decomposition as shown in Figure~\ref{fi:channeldec}~(b) $S_c=\{C_1,C_2,C_3,C_4\}$, where: $C_1=\{0,1,4,5,12,13,15\}$; $C_2=\{0,3,7,11,15\}$; $C_3=\{0,2,6,10,14,15\}$; $C_4=\{0,8,9,15\}$.  Therefore according to Algorithm kD-Draw it can be drawn as a 4-dimensional dominance drawing, which is the best that can be done by using only the channel decomposition approach.  However, after using the module decomposition techniques of this section, we are able to reduce the number of required dimensions to two, as shown Figures~\ref{fi:final_1} and~\ref{fi:final_2}.  Therefore, the use of module decomposition allows us to "group" parts of the graph so that the total number of dimensions is cut in half for this example graph.

\begin{figure}
	\centering
	{\includegraphics[width=1\linewidth]{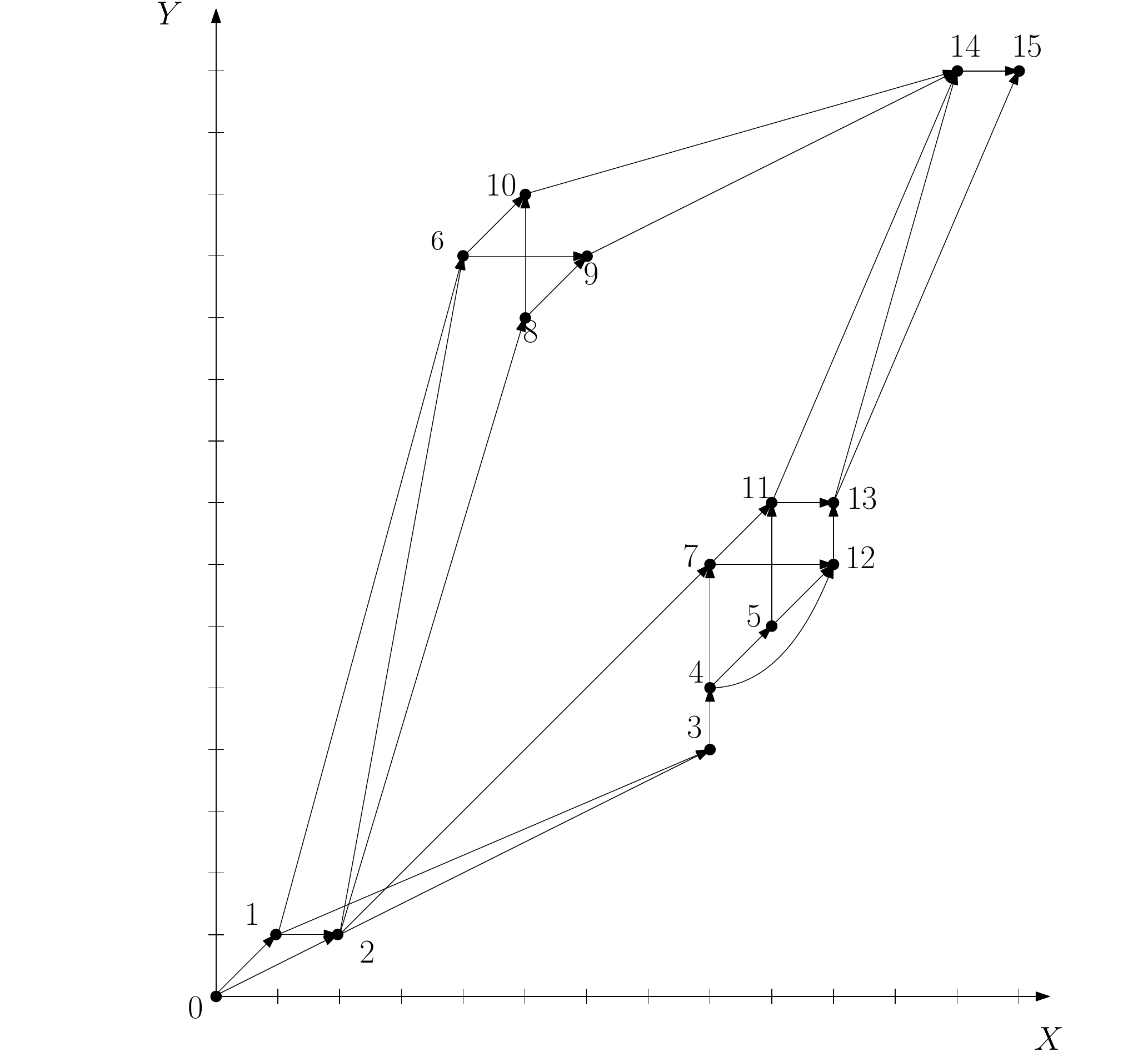}}
	\caption{The drawing of Figure~\ref{fi:final_1} where the added sources and sinks are removed and the original edges of the graph are restored.}
	\label{fi:final_2}
\end{figure}

\indent
\newpage
We are ready now to state the main result of this section:
\begin{theorem}
	Let $G$ be an st-graph and $C_P=\{M_1,...,M_h\}$ be a congruence partition of $G$. Let $G(C_P)=\{G_0, G_1,...,G_h\}$ be the set of all the module-induced graphs of $G$ augmented with the quotient graph $G_0$. Let $w_N$ be the dimensional neck of $G(C_P)$. Algorithm ND-Draw computes a $w_N$-dimensional dominance drawing $\Gamma$ of $G$ in $O(w_{\rho}\rho^2+w_N(h^2+n))$ time.
\end{theorem} 
\begin{proof}
	First, we prove the correctness of the algorithm, i.e., we prove that the computed drawing is a dominance drawing. Then we prove the time complexity of the algorithm. Let $u,v$ be two vertices, where $u\in G_i$ and $v\in G_j$. We need to prove that $r(u,v)=yes\Leftrightarrow u \preceq v$ for any choice of $u$ and $v$. We have $r(u,v)=yes\Leftrightarrow u\preceq v$ in $G$ if $j=i$, since the relative position of $u$ and $v$ in $\Gamma$ do not change in $\Gamma_i$, which is a dominance drawing of $G_i$. Suppose $i\not = j$ and $r(u,v)=yes$. In that case $\mu_i$ reaches $\mu_j$ in $G_0$ by Lemma~\ref{lemma:property_quotient_graph} and $\mu_i\preceq \mu_j$ in $\Gamma_0$, since $\Gamma_0$ is a dominance drawing of $G_0$. The vertex $\mu_j$ is shifted with respect of $\mu_i$ in any dimension $D_g$ in $\Gamma_0$ by $D_g(u,\Gamma_i)\le D_g(t_i,\Gamma_i)$ positions. Moreover, $D_g(u,\Gamma_0)=D_g(\mu_i,\Gamma_0)+D_g(u,\Gamma_i)\le D_g(\mu_j,\Gamma_0)\le D_g(v,\Gamma_0)$, which implies that $u \preceq v$. Now suppose $i\not = j$ and $r(u,v)=no$. In this case, there exists at least one dimension $D_x$ for which $D_x(\mu_i)>D_x(\mu_j)$ by Lemma~\ref{lemma:same_position} and Lemma~\ref{lemma:property_quotient_graph}.  The vertex $\mu_i$ is not shifted with respect of $\mu_j$ in $D_x$. On the contrary, $\mu_j$ is shifted with respect of $\mu_i$ by $D_x(t_i,\Gamma_i)$ positions. Hence,  $D_x(\mu_i)>D_x(\mu_j)$ also after the shifting and in the final drawing. The computational time required has three contributions: $O(w_{\rho}\rho^2)$ time is needed due to Line 1 and Lemma~\ref{lemma:timecomplexity_drawings_computation}; $O(w_nh^2)$ time is needed due to Line 2 and Lemma~\ref{lemma:shiftingtime}; $O(w_Nn)$ time is needed since every vertex is placed in $w_N$ dimensions.
	
\end{proof}

The following corollaries are similar to the ones presented in Subsection \ref{subsection:corollaries}, but they give improved bounds since the dimensional neck $w_N$ is potentially much better than the width $w_{G}$ of any DAG $G$:
\begin{corollary}
	\label{corollary:final1}
	Let $G$ be an st-graph (or DAG) having $n$ vertices and a congruence partition $C_P$. If $w_N$ is the dimensional neck of $G(C_P)$, then: $d_G\le min(\frac{n}{2},w_N)$
\end{corollary}
\begin{corollary}
	\label{corollary:2D1}
	Every DAG $G$ having a congruence partition $C_P$ with dimensional neck equal to 2 has a 2-dimensional dominance drawing.
\end{corollary}
If we define a congruence partition and its dimensional neck for a partial order $P$ as we did for the DAGs, we have the following corollary:
\begin{corollary}
	\label{corollary:partialorder1}
	For any partial order $P$ having $n$ elements and any congruence partition of it. If $w_N$ is the dimensional neck of $C_P$, then: $d_P\le min(\frac{n}{2},w_N)$
\end{corollary}

\section{Conclusions and Open Problems}
\label{Section:Conclusion}
The contributions of our paper are as follows: we proved that every DAG $G$ of width 2 has a 2-dimensional dominance drawing; we proved that $min(\frac{n}{2},w_G)$ is an upper bound of the dominance dimension $d_G$ of the DAG $G$; we proved that $min(\frac{n}{2},w_P)$ is an upper bound for the dimension $d_P$ of a partial order $P$; we introduced the dimensional neck $w_N$ of a DAG $G$, we proved that $w_N\le w_G$ and we improved the upper bounds described above using this new parameter.  This clearly visible by the fact that the DAG of Figure 1 requires four dimensions due to the given channel decomposition, whereas it requires only two dimensions after applying the module decomposition techniques.

An interesting open problem is to reduce the number of dimensions of the drawings computed by the algorithms described in this paper in order to have a weak dominance drawing with a bounded number of fips. In this direction approximation algorithms as well as heuristic algorithms should be developed. Another open problem could be the extension of the family of DAGs having a 2-dimensional dominance drawing with DAGs having width higher than 2. We propose to study the problem of computing a congruence partition having a minimum dimensional neck, since it let us compute dominance drawing with less dimensions in less time. It could be interesting to better understand the relationship between this problem and the NP-hard problem of computing the dominance dimension of a DAG. Additionally, we propose to study the problem of computing a congruence partition having small size ($h$), a small $w_{\rho}$, and a small $\rho$, since the time complexity of ND-Draw depends also on these parameters.  Finally, perhaps the most interesting open problem is to find families of graphs and prove a mathematical relationship between the number of dimensions ($w_G$) required by Algorithm kD-Draw and the number of dimensions required by Algorithm ND-draw, similar to the one that we observed for the example graph of Figure~\ref{fi:channeldec} shown in Figures~\ref{fi:kd} and~\ref{fi:final_2}.  \\\\\\\\\\\\


\noindent
\textbf{Acknowledgement:} We thank Roberto Tamassia for useful discussions and for pointing our attention to Hiraguchi's results.

\newpage

\bibliography{Literature}
\end{document}